\RequirePackage{fix-cm}
\documentclass[smallextended,natbib,runningheads]{svjour3}
\usepackage{amsfonts,graphicx}
\usepackage{amsmath}
\usepackage{amssymb}
\usepackage{hyperref}
\hypersetup{citecolor=blue,colorlinks=true}

\title{A Robust Wald-type Test for Testing the Equality of Two Means from Log-Normal Samples}

\titlerunning{A Wald-type Test for the Equality of Two Means}

  \author{Ayanendranath Basu \and
    Abhijit Mandal \and
    Nirian Mart\'{\i}n \and
    Leandro Pardo}
 \institute{A. Basu \at
    Interdisciplinary Statistical Research Unit, Indian Statistical Institute,  India \and
    A. Mandal \at
    Department of Mathematics, Wayne State University,   USA\\
    Email: abhijit.mandal@wayne.edu\and
    N. Mart\'{\i}n \at
    Department of Statistics, Carlos III University of Madrid,  Spain \and
        L. Pardo\at
    Department of Statistics and O.R., Complutense University of Madrid, Spain
 }
 \authorrunning{Basu, Mandal, Mart\'{\i}n and Pardo}
 \date{}

\begin{document}

\maketitle

\begin{abstract}
The log-normal distribution is one of the most common distributions used for modeling skewed and positive data. It frequently arises in many disciplines of science, specially in the biological and medical sciences. 
The statistical analysis for comparing the means of two independent log-normal
distributions is an issue of significant interest.   In this paper we present a
robust test for this problem. The unknown parameters of the model are estimated
by minimum density power divergence estimators (\citealp{MR1665873},
\textit{Biometrika}, 85(3), 549--559). The robustness as well as the asymptotic properties of the proposed test statistics are rigorously established. The performance of the test is explored
through simulations and real data analysis. The test is compared with some existing methods, and it is demonstrated that the proposed test outperforms the others  in the presence of outliers.

\end{abstract}




\section{Introduction\label{sec1}}

The normal distribution is the most common model used to describe the random variation in real data in many scientific disciplines; the well-known
bell-shaped curve can easily be characterized and described by two parameters: the
mean and the standard deviation. However, often random measurements exhibit skewed patterns which may not be appropriately modeled by the symmetric bell shaped curve.  Skewed distributions are particularly common when
mean values are high, variances large, and values cannot be negative, as is
the case, for example, with species abundance, lengths of latent periods of
infectious diseases, survival time of cancer patients, biological oncology markers sensitivity to fungicides in populations, appearance of lung cancer in cigarette smokers, rainfalls in meteorology, sizes of incomes in economics, etc. Such skewed
distributions are often closely described by the class of  log-normal models, and we are often interested in testing the equality of
the means of two distributions under the log-normal model. In medicine, for example, one may be interested in comparing the latent periods of two infectious diseases, or hospital charges for two groups of patients; in biology this may involve the comparison of the abundance of fish and plankton populations; in environmental science, the distributions of particles, chemicals and organisms in the environment may have to be compared; in linguistics such comparisons may involve the  number of letters per word  and the  number of words per sentence; and in economics one could compare the age of marriage, farm size and income.


For testing the equality of means of two log-normal distributions it is very common to consider a log-transformation of the
data,  apply the $t$-test or the Wilcoxon test to these data,  and to
report the resulting $p$-values for the null hypothesis based on the original
data. This procedure is not convenient, in particular, if the variances
of the two populations are different because in this case testing the equality
of the two means of the log-transformed data are not equivalent to testing the
means of the original data. In this context \cite{MR0038612}, \cite{MR0130756}, \cite{MR0166871}, \cite{MR0395032} and \cite{zhou1997methods} pointed out that the Wilcoxon test as well as the $t$-test can have type I error rates which are very different from the nominal levels when variances of the two populations are not equal. \cite{MR1186261} demonstrated that the Box-Cox transformed two-sample $t$-test based on log-data is asymptotically more powerful than the ordinary $t$-test. \cite{MR1204372} further developed this test to get an improved asymptotic power as well as a better small sample performance.

\cite{zhou1997methods} presented a large sample test and showed, using a
simulation study, that the likelihood-based approach is the best in terms of
the type I error rate and power. \cite{krishnamoorthy2003inferences} demonstrated  that
the distribution of the $Z$ statistic given by \cite{zhou1997methods} is skewed when the
sample sizes are small, or when the parameters are not close to each other.
\cite{gupta2006statistical} considered a score test and compared it with the test
considered by \cite{zhou1997methods}. \cite{MR3334548} considered a test based on the computational approach introduced by \cite{MR2393699}. \cite{ahmed2001test} developed a large sample test, \cite{guo2000testing} proposed three approximation methods, and \cite{MR2750888} considered a
generalized $p$-value approach for comparing the means of several log-normal populations.

In this paper we have developed a class of robust Wald-type test
statistics for testing the equality of the means associated with two
log-normal populations when the parameters are estimated using the minimum
density power divergence (MDPD) estimator. The MDPD estimator was introduced by \cite{MR1665873}. Wald-type test
statistics based on MDPD estimators have been
proposed and studied in \cite{MR3435166} and \cite{ghosh2016influence}. In those
papers the robustness of the Wald-type test statistics based on the MDPD estimator
was established from a theoretical as well as from an applied point of view.

Let $X$ and $Y$ be independent random variables whose distributions are
modeled as log-normals, i.e.,  both $\log X$ and $\log Y$ are normally distributed. Formally, 
\[
\log X\sim\mathcal{N(\mu}_{1},\sigma_{1}^{2})\text{ and }\log Y\sim
\mathcal{N(\mu}_{2},\sigma_{2}^{2})\text{.}%
\]
It is well-known that
\[
E\left[  X\right]  =\exp\left(  \mathcal{\mu}_{1}+\frac{\sigma_{1}^{2}}%
{2}\right)  \text{ and }E\left[  Y\right]  =\exp\left(  \mathcal{\mu}%
_{2}+\frac{\sigma_{2}^{2}}{2}\right)  \text{.}%
\]
It is clear that when $\sigma_{1}^{2}=\sigma_{2}^{2}$,  testing
\begin{equation}
H_{0}:E\left[  X\right]  =E\left[  Y\right]  \text{ against }H_{1}:E\left[
X\right]  \neq E\left[  Y\right]  \text{ } \label{1}%
\end{equation}
is equivalent to testing
\begin{equation}
H_{0}:\mu_{1}=\mu_{2}\text{ against }H_{1}:\mu_{1}\neq\mu_{2}. \label{2}%
\end{equation}
But if $\sigma_{1}^{2}\neq\sigma_{2}^{2}$, the problem of testing for the
hypotheses in (\ref{1}) cannot be solved by testing the hypotheses in
(\ref{2}). For more details see \cite{zhou1997methods}.

In Section \ref{sec2} we introduce some notation in relation to the MDPD
estimator and we present the asymptotic distribution of the MDPD estimator for
$\mu_{1}$ and $\sigma_{1}$ in a log-normal model as well as the influence function associated with the MDPD estimator. The Wald-type test
statistics are introduced and their asymptotic distributions are studied in
Section \ref{sec3}. The robustness properties of the Wald-type test statistics
considered in this paper are studied in Section \ref{sec4}. In Section \ref{sec6} a simulation study is developed, and   some numerical examples are considered in Section \ref{sec5}. Finally, some conclusions are presented in Section \ref{conclusion}.

\section{The minimum density power divergence estimator: asymptotic
properties\label{sec2}}

For any two probability density functions $f$ and $g$, the density power
divergence measure is defined, as the function of a single tuning parameter
$\beta\geq0$, as
\begin{equation}
d_{\beta}(g,f)=\left\{
\begin{array}
[c]{ll}%
\int\left\{  f^{1+\beta}(z)-\left(  1+\frac{1}{\beta}\right)  f^{\beta
}(z)g(z)+\frac{1}{\beta}g^{1+\beta}(z)\right\}  dz, & \text{for}%
\mathrm{~}\beta>0,\\[2ex]%
\int g(z)\log\left(  \frac{g(z)}{f(z)}\right)  dz, & \text{for}\mathrm{~}%
\beta=0.
\end{array}
\right.  \label{3}%
\end{equation}
Let $Z$ be distributed with density function $f_{\boldsymbol{\theta}}(z)$ with
respect to some $\sigma$-finite measure (usually Lebesgue measure or counting
measure) where $\boldsymbol{\theta}\in\Theta\subset\mathbb{R}^{p}$. We are
interested in considering an estimator for $\boldsymbol{\theta}$ based on
(\ref{3}). We shall denote by $G$ the distribution function corresponding to
the density function $g$ that generates the data. The MDPD functional at $G$, denoted by $T_{\beta}\left(  G\right)  ,$ is
defined as
\[
d_{\beta}\left(  g,f_{T_{\beta}\left(  G\right)  }\right)  =\min
_{\boldsymbol{\theta}\in\Theta}d_{\beta}\left(  g,f_{\boldsymbol{\theta}%
}\right)  .
\]
\cite{MR1665873} considered the MDPD estimator of $\boldsymbol{\theta}$ given by
\[
\widehat{\boldsymbol{\theta}}_{\beta}=\boldsymbol{T}_{\beta}\left(
G_{n}\right)  ,
\]
where $G_{n}$ is the empirical distribution function associated with a random
sample $Z_{1},...,Z_{n}$ from the population with density function $g$. It is
easy to see that
\[
\widehat{\boldsymbol{\theta}}_{\beta}=\arg\min_{\boldsymbol{\theta}\in\Theta
}\left(  \int f_{\boldsymbol{\theta}}^{1+\beta}(z)dz-\left(  1+\frac{1}{\beta
}\right)  \frac{1}{n}\sum\limits_{i=1}^{n}f_{\boldsymbol{\theta}}^{\beta
}(Z_{i})\right)
\]
for $\beta>0$, and
\[
\widehat{\boldsymbol{\theta}}_{\beta=0}=\arg\min_{\boldsymbol{\theta}\in
\Theta}\left(  -\frac{1}{n}\sum\limits_{i=1}^{n}\log f_{\boldsymbol{\theta}%
}(Z_{i})\right)
\]
for $\beta=0$. We can see that in the latter case we get the maximum likelihood estimator (MLE) as the solution.

Now, we shall focus on a log-normal population $X$ with
unknown parameter  $\boldsymbol{\theta}=(\mu_{1},\sigma_{1})^{T}$. For
population $Y$ with unknown parameter  $\boldsymbol{\theta}=(\mu
_{2},\sigma_{2})^{T}$, the derivations would be exactly the same. Let
$X_{1},X_{2},\ldots,X_{n_{1}}$ be a random sample of size $n_{1}$ from a
log-normal population $\log X\sim\mathcal{N}(\mu_{1},\sigma_{1}^{2})$, where
both parameters are unknown. The pair $\mu_1^*,\sigma_1^*$ denotes the
true value of $\mu_{1},\sigma_{1}$. Let $f_{\mu_{1},\sigma_{1}}(x)$
represent the density function of a log-normal variable with parameters $\mu_1$ and $\sigma_1$. For a given $\beta$,
we get the MDPD estimators $\widehat{\mu}_{1,\beta}$ and $\widehat{\sigma
}_{1,\beta}$ of $\mu_{1}$ and $\sigma_1$ by minimizing the  function 
\begin{equation}
\int_{\mathbb{R}}f_{\mu_{1},\sigma_{1}}^{1+\beta}(x)dx-\left(  1+\frac
{1}{\beta}\right)  \frac{1}{n_{1}}\sum_{i=1}^{n_{1}}f_{\mu_{1},\sigma_{1}%
}^{\beta}(X_{i}),\text{\qquad for }\beta>0,\label{4}%
\end{equation}
and
\begin{equation}
-\frac{1}{n_{1}}\sum_{i=1}^{n_{1}}\log f_{\mu_{1},\sigma_{1}}(X_{i}%
),\text{\qquad for }\beta=0,\label{5}%
\end{equation}
over $\mu_1$ and $\sigma_1$. 
For $\beta=0$, the objective function in (\ref{5}) is the negative of the
usual log likelihood and has the classical MLE as the
minimizer. For a log-normal density, the function in (\ref{4}) simplifies to $(1+\beta) h_{n_{1},\beta}(\mu_{1},\sigma_{1})$, where
\[
h_{n_{1},\beta}(\mu_{1},\sigma_{1})=\frac{1}{\sigma_{1}^{\beta}(2\pi
)^{\frac{\beta}{2}}}\left\{  \frac{\exp\left(  -\beta\mu_{1}+\frac{\sigma
_{1}^{2}\beta^{2}}{2\left(  1+\beta\right)  }\right)  }{\left(  1+\beta
\right)  ^{3/2}}-\frac{1}{n_{1}\beta}\sum_{i=1}^{n_{1}}\frac{1}{X_{i}^{\beta}%
}\exp\left(  -\frac{1}{2}\left(  \frac{\log X_{i}-\mu_{1}}{{\sigma}_{1}}\right)
^{2}\beta\right)  \right\}  .
\]
In order to get $\widehat{\mu}_{1,\beta}$ and $\widehat{\sigma}_{1,\beta}$, we
have to solve the estimating equation
\begin{equation}
\mathbf{h}_{n_{1},\beta}^{\prime}(\widehat{\mu}_{1,\beta},\widehat{\sigma
}_{1,\beta})=%
\begin{pmatrix}
_{1}h_{n_{1},\beta}^{\prime}(\widehat{\mu}_{1,\beta},\widehat{\sigma}%
_{1,\beta})\\
_{2}h_{n_{1},\beta}^{\prime}(\widehat{\mu}_{1,\beta},\widehat{\sigma}%
_{1,\beta})
\end{pmatrix}
=\boldsymbol{0}_{2},\label{6}%
\end{equation}
where
\begin{equation}
_{1}h_{n_{1},\beta}^{\prime}(\widehat{\mu}_{1,\beta},\widehat{\sigma}%
_{1,\beta})=\left.  \frac{\partial h_{n_{1},\beta}(\mu_{1},\widehat{\sigma
}_{1,\beta})}{\partial\mu_{1}}\right\vert _{\mu_{1}=\widehat{\mu}_{1,\beta}%
},\qquad_{2}h_{n_{1},\beta}^{\prime}(\widehat{\mu}_{1,\beta},\widehat{\sigma
}_{1,\beta})=\left.  \frac{\partial h_{n_{1},\beta}(\widehat{\mu}_{1,\beta
},\sigma)}{\partial\sigma}\right\vert _{\sigma=\widehat{\sigma}_{1,\beta}%
},\label{7}%
\end{equation}
and $\boldsymbol{0}_{2}$ represents a zero vector of length $2$.

Using a Taylor series expansion of the function in Equation (\ref{6}), it is
easy to show that
\begin{align}
\sqrt{n_{1}}%
\begin{pmatrix}
\widehat{\mu}_{1,\beta}-\mu_1^*\\
\widehat{\sigma}_{1,\beta}-\sigma_1^*%
\end{pmatrix}
&  =\sqrt{n_{1}}\mathbf{H}_{n_{1},\beta}^{-1}(\mu_1^*,\sigma_1^*%
)\boldsymbol{h}_{n_{1},\beta}^{\prime}(\mu_1^*,\sigma_1^*)+o_{p}%
(1)\nonumber\\
&  =\sqrt{n_{1}}\mathbf{J}_{\beta}^{-1}(\mu_1^*,\sigma_1^*)\boldsymbol{h}%
_{n_{1},\beta}^{\prime}(\mu_1^*,\sigma_1^*)+o_{p}(1),\label{muSigma}%
\end{align}
where%
\[
\mathbf{H}_{n_{1},\beta}(\mu_1^*,\sigma_1^*)=\left.
\begin{pmatrix}
\frac{\partial^{2}h_{n_{1},\beta}(\mu_{1},\sigma_{1})}{\partial\mu_{1}^{2}} &
\frac{\partial^{2}h_{n_{1},\beta}(\mu_{1},\sigma_{1})}{\partial\mu_{1}%
\partial\sigma_{1}}\\
\frac{\partial^{2}h_{n_{1},\beta}(\mu_{1},\sigma_{1})}{\partial\mu_{1}%
\partial\sigma_{1}} & \frac{\partial^{2}h_{n_{1},\beta}(\mu_{1},\sigma_{1}%
)}{\partial\sigma_{1}^{2}}%
\end{pmatrix}
\right\vert _{\mu_{1}=\mu_1^*,\sigma_{1}=\sigma_1^*},
\]%
\begin{align}
\boldsymbol{J}_{\beta}(\mu_1^*,\sigma_1^*) &  =\lim_{n_{1}\rightarrow
\infty}\mathbf{H}_{n_{1},\beta}(\mu_1^*,\sigma_1^*)\label{9} \nonumber\\
&  =\mathcal{L}\left(  \beta,\mu_1^*,\sigma_1^*\right)  \nonumber\\
& \times \left(
\begin{array}
[c]{cc}%
\frac{1+\beta+\beta^{2}\sigma_1^{*2}}{\sigma_1^*} & \beta\left(
-\frac{\beta^{2}\sigma_1^{*2}}{1+\beta}+\beta-2\right)  \\
\beta\left(  -\frac{\beta^{2}\sigma_1^{*2}}{1+\beta}+\beta-2\right)  
&
\frac{1}{\sigma_1^*}\left(  \frac{\beta^{4}\sigma_1^{*4}}{\left(
1+\beta\right)  ^{2}}+\frac{6\beta^{2}\sigma_1^{*2}}{1+\beta}+\beta
^{2}\left(  1-2\sigma_1^{*2}\right)  +2\right)
\end{array}
\right)  ,\nonumber
\end{align}
and%
\[
\mathcal{L}\left(  \beta,\mu_1^*,\sigma_1^*\right)  =\frac{\exp\left(
-\beta\mu_1^*+\frac{\beta^{2}\sigma_1^{*2}}{2\left(  1+\beta\right)
}\right)  }{\sigma_1^{*1+\beta}(2\pi)^{\frac{\beta}{2}}\left(
1+\beta\right)  ^{5/2}}.
\]
Applying the Central Limit Theorem and after some algebra it is not difficult
to see that
\begin{equation}
\sqrt{n_1}\mathbf{h}_{n_{1},\beta}^{\prime}(\mu_1^*,\sigma_1^*)\underset{n_{1}%
\rightarrow\infty}{\overset{\mathcal{L}}{\longrightarrow}}\mathcal{N}\left(
\boldsymbol{0}_{2},\boldsymbol{K}_{\beta}(\mu_1^*,\sigma_1^*)\right)
,\label{10} \nonumber%
\end{equation}
where
\begin{align}
\boldsymbol{K}_{\beta}(\mu_1^*,\sigma_1^*) &  =\mathcal{L}^{\ast}\left(
\beta,\mu_1^*,\sigma_1^*\right)  \nonumber\\
& \times \left(
\begin{array}
[c]{cc}%
\frac{1+2\beta+4\beta^{2}\sigma_1^{*2}}{\sigma_1^*} & \beta\left(
\frac{-8\sigma_1^{*2}\beta^{2}}{1+2\beta}+4\beta-4\right)  \\
\beta\left(  \frac{-8\sigma_1^{*2}\beta^{2}}{1+2\beta}+4\beta-4\right)   &
\frac{1}{\sigma_1^*}\left(  \frac{16\beta^{4}\sigma_1^{*4}}{\left(
1+2\beta\right)  ^{2}}+\frac{24\sigma_1^{*2}\beta^{2}}{1+2\beta}+4\beta
^{2}\left(  1-2\sigma_1^{*2}\right)  +2\right)
\end{array}
\right)  \nonumber\\
&  -\mathcal{L}^{\ast\ast}\left(  \beta,\mu_1^*,\sigma_1^*\right)  \left(
\begin{array}
[c]{cc}%
\sigma_1^{*2}\beta^{2} & -\frac{\sigma_1^*\beta^{2}\left(  -1-\beta
+\beta\sigma_1^{*2}\right)  }{1+\beta}\\
-\frac{\sigma_1^*\beta^{2}\left(  -1-\beta+\beta\sigma_1^{*2}\right)
}{1+\beta} & \frac{\beta^{2}\left(  -1-\beta+\beta\sigma_1^{*2}\right)
^{2}}{\left(  1+\beta\right)  ^{2}}%
\end{array}
\right)  .\label{11} \nonumber%
\end{align}
Here%
\[
\mathcal{L}^{\ast}\left(  \beta,\mu_1^*,\sigma_1^*\right)  =\frac
{\exp\left(  -2\beta\mu_1^*+2\frac{\beta^{2}\sigma_1^{*2}}{1+2\beta
}\right)  }{\sigma_1^{*2\beta+1}(2\pi)^{\beta}\left(  1+2\beta\right)
^{5/2}}\text{ and }\mathcal{L}^{\ast\ast}\left(  \beta,\mu_1^*,\sigma
_1^*\right)  =\frac{\exp\left(  -2\beta\mu_1^*+\frac{\beta^{2}\sigma
_1^{*2}}{1+\beta}\right)  }{\sigma_1^{*2\beta+2}(2\pi)^{\beta}\left(
1+\beta\right)  ^{3}}.
\]
Finally, by (\ref{muSigma})%
\begin{equation}
\sqrt{n_{1}}%
\begin{pmatrix}
\widehat{\mu}_{1,\beta}-\mu_1^*\\
\widehat{\sigma}_{1,\beta}-\sigma_1^*%
\end{pmatrix}
\underset{n_{1}\rightarrow\infty}{\overset{\mathcal{L}}{\longrightarrow}%
}\mathcal{N}\left(  \boldsymbol{0}_{2},\boldsymbol{\Sigma}_{\beta
}(\mathcal{\mu}_1^*,\sigma_1^*)\right)  ,\label{as}%
\end{equation}
where%
\[
\boldsymbol{\Sigma}_{\beta}(\mathcal{\mu}_1^*,\sigma_1^*)=\boldsymbol{J}%
_{\beta}^{-1}(\mu_1^*,\sigma_1^*)\boldsymbol{K}_{\beta}(\mu_1^*%
,\sigma_1^*)\boldsymbol{J}_{\beta}^{-1}(\mu_1^*,\sigma_1^*).
\]

The influence function ($IF$) introduced by \cite{MR2617979,MR0362657}  indicates how an
infinitesimal proportion of contamination affects the estimate in large
samples. Formally, the $IF$ gives a quantitative expression of the change in the
estimate that results from perturbing the underlying distribution, $F$, by a
point mass at a certain location and it is the most useful heuristic tool of
robust statistics. \textquotedblleft The $IF$ is mainly a heuristic tool, with
an intuitive interpretation\textquotedblright\ (\citealt{MR829458}, p. 83).

We denote by $F_{\mu_1^*,\sigma_1^*}$ the distribution function of a
log-normal distribution with parameters $\mu_1^*$ and $\sigma_1^*$. For
the MDPD functional, $\boldsymbol{T}_{\beta}^{1}$,
the $IF$, at $F_{\mu_1^*,\sigma_1^*},$ can be expressed as
\[
\mathcal{IF}\left(  x,\boldsymbol{T}_{\beta}^{1},F_{\mu_1^*,\sigma_1^*%
}\right)  =\lim_{\varepsilon\rightarrow0}\frac{\boldsymbol{T}_{\beta}%
^{1}((1-\varepsilon)F_{\mu_1^*,\sigma_1^*}+\varepsilon\delta
_{x})-\boldsymbol{T}_{\beta}^{1}(F_{\mu_1^*,\sigma_1^*})}{\varepsilon},
\]
where $\delta_{x}$ is the distribution function corresponding to the degenerate
distribution at $x$. It measures the normalized asymptotic bias caused by an
infinitesimal contamination at point $x$ in the observations.
Hence, the $IF$ reflects the bias caused by adding a few outliers at the point,
standardized by the amount $\varepsilon$ of contamination. Therefore, a
bounded $IF$ leads to robust estimators. Note that this kind of differentiation
of statistical functionals is a differentiation in the sense of von Mises.

In \cite{MR1665873} the influence function for the MDPD
functional is given as
\begin{equation}
\mathcal{IF}\left(  x,\boldsymbol{T}_{\beta}^{1},F_{_{\mu_1^*,\sigma_1^*}%
}\right)  =\boldsymbol{J}_{\beta}^{-1}\left(  \mu_1^*,\sigma_1^*\right)
\left(  \boldsymbol{u}_{\mu_1^*,\sigma_1^*}(x)f_{\mu_1^*,\sigma_1^*%
}^{\beta}(x)-\boldsymbol{\xi}\left(  \mu_1^*,\sigma_1^*\right)  \right)
,\label{IF}%
\end{equation}
where $\boldsymbol{u}_{\mu_1^*,\sigma_1^*}(x)$ is the score function
given, for the log-normal model, by
\[
\boldsymbol{u}_{\mu_1^*,\sigma_1^*}(x)=\left(
\begin{array}
[c]{c}%
\frac{\log x-\mu_1^*}{\sigma_1^*}\\
\frac{1}{\sigma_1^*}\left(  \left(  \frac{\log x-\mu_1^*}{\sigma_1^*%
}\right)  ^{2}-1\right)
\end{array}
\right)
\]
and
\begin{align*}
\boldsymbol{\xi}\left(  \mu_1^*,\sigma_1^*\right)   &  =\int_{0}^{\infty
}\boldsymbol{u}_{\mu_1^*,\sigma_1^*}(x)f_{\mu_1^*,\sigma_1^*}%
^{1+\beta}(x)dx\\
&  =\frac{1}{\sigma_1^{*1+\beta}\left(  2\pi\right)  ^{\frac{\beta}{2}%
}\left(  1+\beta\right)  ^{\frac{3}{2}}}\exp\left(  -\mu_1^*\beta
+\frac{\sigma_1^{*2}\beta^{2}}{2\left(  1+\beta\right)  }\right)  \left(
\begin{array}
[c]{c}%
-\sigma_1^*\beta\\
\frac{\beta\left(  1-\beta+\beta\sigma_1^{*2}\right)  }{1+\beta}%
\end{array}
\right)  .
\end{align*}
After some fairly simple  algebra it is possible to see that Equation (\ref{IF}) is a bounded function
of $x$ for $\beta>0$. For $\beta=0$, which generated the MLEs, the influence function is given by 
\[
\mathcal{IF}\left(  x,\boldsymbol{T}_{\beta=0}^{1},F_{_{\mu_1^*,\sigma
_1^*}}\right)  =%
\begin{pmatrix}
\log x-\mu_1^*\\
-\frac{1}{2\sigma_1^*}\left(  \sigma_1^{*2}-\mu_1^{*2}+2\mu_1^*\log
x-\log^{2}x\right)
\end{pmatrix}
,
\]
which is unbounded. In terms of the influence function, therefore,  the MLE is a non-robust
estimator and the MDPD estimators are robust
estimators for $\beta>0$.

In Figures \ref{fig:lp1} and \ref{fig:lp2}, we present respectively the first
and second components of $\mathcal{IF}\left(  x,\boldsymbol{T}_{\beta}%
^{1},F_{_{\mu_1^*,\sigma_1^*}}\right)  $, when $\mu_1^*=0$
and$\ \sigma_1^*=1$, for $\beta=0$, $0.3$, $0.5$\ and $1$. The curves eventually become stable at all cases except $\beta = 0$.

\begin{figure}
\centering
\begin{tabular}
[c]{c}${{\includegraphics[height=2.5in,width=4.2704in]{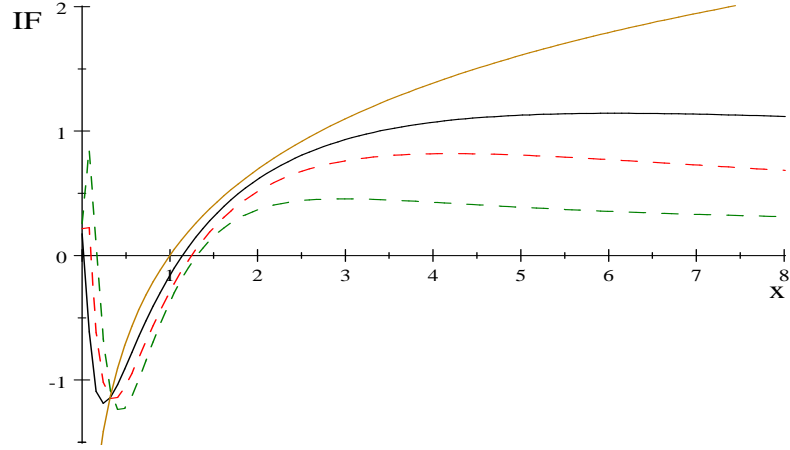}}}$%
\end{tabular}
\caption{First component of $\mathcal{IF}\left(  x,\boldsymbol{T}_{\beta}^{1},F_{_{\mu_1^*,\sigma_1^*}}\right)  $ with $\mu_1^*=0\ $and $\sigma
_1^*=1$; lines $\beta=0$ (yellow), $\beta=0.3$ (black), $\beta=0.5$ (red) and $\beta=1$ (green).}
\label{fig:lp1}
\end{figure}
\begin{figure}
\centering
\begin{tabular}[c]{c}%
${{\includegraphics[height=2.5in,width=4.2704in]{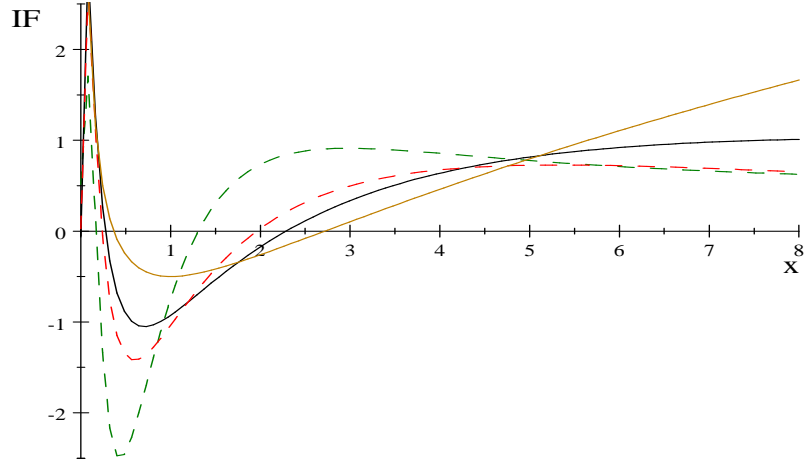}}}$%
\end{tabular}
\caption{Second component of $\mathcal{IF}\left(  x,\boldsymbol{T}_{\beta}^{1},F_{_{\mu_1^*,\sigma_1^*}}\right)  $ with $\mu_1^*=0\ $and $\sigma_1^*=1$; lines $\beta=0$ (yellow), $\beta=0.3$ (black), $\beta=0.5$ (red) and $\beta=1$ (green).}
\label{fig:lp2}
\end{figure}

\section{Wald-type test statistics based on MDPD estimators\label{sec3}}

In this section we will present a procedure based on a Wald-type
statistic for testing the hypothesis in  (\ref{1}) using the MDPD estimator. Our hypotheses of
interest are
\begin{equation}
H_{0}:\exp(\mathcal{\mu}_{1}+\tfrac{\sigma_{1}^{2}}{2})=\exp(\mathcal{\mu}%
_{2}+\tfrac{\sigma_{2}^{2}}{2})\text{ against }H_{1}:\exp(\mathcal{\mu}%
_{1}+\tfrac{\sigma_{1}^{2}}{2})\neq\exp(\mathcal{\mu}_{2}+\tfrac{\sigma
_{2}^{2}}{2}). \label{A}%
\end{equation}
In this case, the parameter space is given by
\[
\Lambda=\left\{  \boldsymbol{\eta}=\left(  \mathcal{\mu}_{1},\sigma
_{1},\mathcal{\mu}_{2},\sigma_{2}\right)  ^{T} : \mathcal{\mu}_{i}%
\in\mathbb{R}\text{ and }\sigma_{i}\in\mathbb{R}^{+},\text{ }i=1,2\right\}
\]
and if we denote
\[
m\left(  \boldsymbol{\eta}\right)  =\exp(\mathcal{\mu}_{1}+\tfrac{\sigma
_{1}^{2}}{2})-\exp(\mathcal{\mu}_{2}+\tfrac{\sigma_{2}^{2}}{2}),
\]
the null hypothesis can be defined by
\[
\Lambda_{0}=\left\{  \boldsymbol{\eta}=\left(  \mathcal{\mu}_{1},\sigma
_{1},\mathcal{\mu}_{2},\sigma_{2}\right)  ^{T}\in\Lambda :  m\left(
\boldsymbol{\eta}\right)  =0\right\}  .
\]
It is clear that
\begin{equation}
\frac{\partial m(\boldsymbol{\eta})}{\partial\boldsymbol{\eta}}=\left(
\exp(\mathcal{\mu}_{1}+\tfrac{\sigma_{1}^{2}}{2}),\sigma_{1}\exp(\mathcal{\mu
}_{1}+\tfrac{\sigma_{1}^{2}}{2}),-\exp(\mathcal{\mu}_{2}+\tfrac{\sigma_{2}%
^{2}}{2}),-\sigma_{2}\exp(\mathcal{\mu}_{2}+\tfrac{\sigma_{2}^{2}}{2})\right)
^{T}. \label{dg}%
\end{equation}

We consider a simple random sample $X_{1},...,X_{n_{1}}$ from population $\log
X\sim\mathcal{N(\mu}_{1},\sigma_{1}^{2})$ and a simple random sample
$Y_{1},...,Y_{n_{2}}$ from population $\log Y\sim\mathcal{N(\mu}_{2}%
,\sigma_{2}^{2})$. We shall assume that the two samples are independent.
We shall denote by $(\widehat{\mathcal{\mu}}_{1,\beta},\widehat{\sigma}%
_{1,\beta})^{T}$ the MDPD estimator based on $X_{1},...,X_{n_{1}}$ and
$(\widehat{\mathcal{\mu}}_{2,\beta},\widehat{\sigma}_{2,\beta})^{T}$ the
MDPD estimator based on $Y_{1},...,Y_{n_{2}}$ (see Section \ref{sec2}, for
more details). In the next theorem we shall establish the asymptotic distribution of
$\widehat{\boldsymbol{\eta}}_{\beta}=(\widehat{\mathcal{\mu}}_{1,\beta
},\widehat{\sigma}_{1,\beta},\widehat{\mathcal{\mu}}_{2,\beta
},\widehat{\sigma}_{2,\beta})^{T}$.

\begin{theorem}
Let
\begin{equation}
w=\lim_{n_{1},n_{2}\rightarrow\infty}\frac{n_{1}}{n_{1}+n_{2}}\in
(0,1)\label{EQ:w}%
\end{equation}
be the limiting proportion of observations from the first population in the
whole sample. Then, the MDPD estimator of $\boldsymbol{\eta}$,
$\widehat{\boldsymbol{\eta}}_{\beta}$, has the asymptotic distribution given
by
\begin{equation}
\sqrt{\frac{n_{1}n_{2}}{n_{1}+n_{2}}}(\widehat{\boldsymbol{\eta}}_{\beta
}-\boldsymbol{\eta}_{0})\underset{n_{1},n_{2}\rightarrow\infty
}{\overset{\mathcal{L}}{\longrightarrow}}\mathcal{N}\left(  \boldsymbol{0}%
_{4},\boldsymbol{\Sigma}_{w,\beta}(\boldsymbol{\eta}_{0})\right)
,\label{eqTh0}%
\end{equation}
where $\boldsymbol{\eta}_{0}=\left(  \mathcal{\mu}_1^*,\sigma_1^*%
,\mathcal{\mu}_2^*,\sigma_2^*\right)  ^{T}$ is the true parameter vector
of $\boldsymbol{\eta}$,%
\begin{equation}
\boldsymbol{\Sigma}_{w,\beta}(\boldsymbol{\eta})=\left(
\begin{array}
[c]{cc}%
\left(  1-w\right)  \boldsymbol{\Sigma}_{\beta}(\mathcal{\mu}_{1},\sigma
_{1}) & \boldsymbol{0}_{2\times2}\\
\boldsymbol{0}_{2\times2} & w\boldsymbol{\Sigma}_{\beta}(\mathcal{\mu}%
_{2},\sigma_{2})
\end{array}
\right)  ,\label{sigMat}%
\end{equation}
with
\[
\boldsymbol{\Sigma}_{\beta}(\mathcal{\mu}_{i},\sigma_{i})=\boldsymbol{J}%
_{\beta}^{-1}(\mathcal{\mu}_{i},\sigma_{i})\boldsymbol{K}_{\beta}%
(\mathcal{\mu}_{i},\sigma_{i})\boldsymbol{J}_{\beta}^{-1}(\mathcal{\mu}%
_{i},\sigma_{i}),\quad i=1,2.
\]
Here, $\boldsymbol{0}_{n}$ is null vector of length $n$, and $\boldsymbol{0}%
_{n\times n}$ is a null matrix of dimension $n\times n$.
\end{theorem}

\begin{proof}
The result follows from the asymptotic distribution for the first population, the convergence in Equation (\ref{as}), the corresponding convergence for the $Y$ population, and the independence of the samples. 
\hfill\rule{0.7em}{0.7em}
\end{proof}

\begin{definition}
We define the Wald-type test statistic based on MDPD estimator for testing hypotheses in (\ref{A}) by
\begin{equation}
W_{n_{1},n_{2}}^{\beta}=\frac{n_{1}n_{2}}{n_{1}+n_{2}}\frac{m^{2}%
(\widehat{\boldsymbol{\eta}}_{\beta})}{\sigma_{\beta,m}^{2}%
(\widehat{\boldsymbol{\eta}}_{\beta})},\label{wald}%
\end{equation}
where%
\begin{equation}
\sigma_{\beta,m}^{2}(\boldsymbol{\eta})=\left(\frac{\partial m(\boldsymbol{\eta}%
)}{\partial\boldsymbol{\eta}} \right)^{T}\boldsymbol{\Sigma}_{w,\beta}%
(\boldsymbol{\eta})\frac{\partial m(\boldsymbol{\eta})}{\partial
\boldsymbol{\eta}}.\label{sigmag}%
\end{equation}

\end{definition}

In the next theorem we shall establish the asymptotic distribution of
$W_{n_{1},n_{2}}^{\beta}$.

\begin{theorem}
Under the conditions of the previous theorem we get, asymptotically, 
\[
W_{n_{1},n_{2}}^{\beta}\underset{n_{1},n_{2}\rightarrow\infty
}{\overset{\mathcal{L}}{\longrightarrow}}\mathcal{\chi}_{1}^{2}.
\]

\end{theorem}

\begin{proof}
It is clear that
\[
\sqrt{\frac{n_{1}n_{2}}{n_{1}+n_{2}}}m(\widehat{\boldsymbol{\eta}}_{\beta
})=\sqrt{\frac{n_{1}n_{2}}{n_{1}+n_{2}}}m\left(  \boldsymbol{\eta}_{0}\right)
+\sqrt{\frac{n_{1}n_{2}}{n_{1}+n_{2}}}\left.  \frac{\partial
m(\boldsymbol{\eta})}{\partial\boldsymbol{\eta}^{T}}\right\vert
_{\boldsymbol{\eta}=\boldsymbol{\eta}_{0}}(\widehat{\boldsymbol{\eta}}_{\beta
}-\boldsymbol{\eta}_{0})+o_{p}(1).
\]
Now, the result follows because $m\left(  \boldsymbol{\eta}_{0}\right)  =0.$
\hfill\rule{0.7em}{0.7em}
\end{proof}

\begin{remark}
In the particular case of $\beta=0$ we will get the MLEs and we have that $\boldsymbol{J}_{\beta=0}(\mu_{i},\sigma_{i}%
)$, $\boldsymbol{K}_{\beta=0}(\mu_{i},\sigma_{i})$ matrices coincide with the
Fisher information matrix%
\[
\boldsymbol{J}_{\beta=0}(\mu_{i},\sigma_{i})=\boldsymbol{K}_{\beta=0}(\mu
_{i},\sigma_{i})=\left(
\begin{array}
[c]{cc}%
\frac{1}{\sigma_{i}^{2}} & 0\\
0 & \frac{2}{\sigma_{i}^{2}}%
\end{array}
\right)  =\boldsymbol{I}_{F}(\mu_{i},\sigma_{i}),\quad i=1,2.
\]
Therefore
\[
\boldsymbol{\Sigma}_{\beta=0}(\mathcal{\mu}_{i},\sigma_{i})=\boldsymbol{I}%
_{F}^{-1}(\mu_{i},\sigma_{i}),\quad i=1,2.
\]
In this situation, the MDPD estimator based Wald-type test statistic,
$W_{n_{1},n_{2}}^{\beta=0}$, coincides with the classical Wald test
\[
\frac{n_{1}n_{2}}{n_{1}+n_{2}}\frac{m^{2}(\widehat{\boldsymbol{\eta}}%
_{\beta=0})}{\sigma_{\beta=0,m}^{2}(\widehat{\boldsymbol{\eta}}_{\beta=0})},
\]
where%
\[
\sigma_{\beta=0,m}^{2}(\widehat{\boldsymbol{\eta}}_{\beta=0})=\left.
\frac{\partial m(\boldsymbol{\eta})}{\partial\boldsymbol{\eta}^{T}}\right\vert
_{\boldsymbol{\eta=}\widehat{\boldsymbol{\eta}}_{\beta=0}}%
\begin{pmatrix}
(1-w)\boldsymbol{I}_{F}^{-1}(\widehat{\mu}_{1,\beta=0},\widehat{\sigma}%
_{1,\beta=0}) & \boldsymbol{0}_{2\times2}\\
\boldsymbol{0}_{2\times2} & w\boldsymbol{I}_{F}^{-1}(\widehat{\mu
}_{2,\beta=0},\widehat{\sigma}_{2,\beta=0})
\end{pmatrix}
\left.  \frac{\partial m(\boldsymbol{\eta})}{\partial\boldsymbol{\eta}%
}\right\vert _{\boldsymbol{\eta=}\widehat{\boldsymbol{\eta}}_{\beta=0}}.
\]

\end{remark}

Based on the previous theorem, we shall reject the null hypothesis in
(\ref{A}) if $W_{n_{1},n_{2}}^{\beta}>\chi_{1,\alpha}^{2}$, with
$\chi_{1,\alpha}^{2}$ being the $100(1-\alpha)$-th percentile point of
chi-squared distribution with $1$ degree of freedom. 

Now, we are going to get
an approximation to the power function for the Wald-type test statistics given
by Equation (\ref{wald}).
We consider a point under the alternative hypothesis of (\ref{A}), i.e.
$\boldsymbol{\eta}^{\ast}=\left(  \mathcal{\mu}_{1}^{\ast},\sigma_{1}^{\ast
},\mathcal{\mu}_{2}^{\ast},\sigma_{2}^{\ast}\right)  ^{T}\in\Lambda
-\Lambda_{0}$. We could then assume that $\boldsymbol{\eta}^{\ast}$ is the
true value of the unknown parameter and
\[
\widehat{\boldsymbol{\eta}}_{\beta}\underset{n_{1},n_{2}\rightarrow
\infty}{\overset{\mathcal{P}}{\longrightarrow}}\boldsymbol{\eta}^{\ast}.
\]

\begin{theorem}
Under the alternative hypothesis in (\ref{A}), we have
\[
\sqrt{\frac{n_{1}n_{2}}{n_{1}+n_{2}}}\left(  V_{\beta}%
(\widehat{\boldsymbol{\eta}}_{\beta},\widehat{\boldsymbol{\eta}}_{\beta
})-V_{\beta}(\boldsymbol{\eta}^{\ast},\boldsymbol{\eta}^{\ast})\right)
\underset{n_{1},n_{2}\rightarrow\infty}{\overset{\mathcal{L}}{\longrightarrow
}}\mathcal{N}\left(  0,\phi_{\beta,m}^{2}\left(  \boldsymbol{\eta}^{\ast
}\right)  \right)
\]
where%
\begin{align*}
V_{\beta}(\boldsymbol{\eta}_{1},\boldsymbol{\eta}_{2}) &  =\frac
{m^{2}(\boldsymbol{\eta}_{1})}{\sigma_{\beta,m}^{2}(\boldsymbol{\eta}_{2})},\\
\phi_{\beta,m}^{2}\left(  \boldsymbol{\eta}^{\ast}\right)   &  =\left.
\frac{\partial V_{\beta}(\boldsymbol{\eta}_{1},\boldsymbol{\eta}^{\ast}%
)}{\partial\boldsymbol{\eta}_{1}^{T}}\right\vert _{\boldsymbol{\eta}%
_{1}=\boldsymbol{\eta}^{\ast}}\boldsymbol{\Sigma}_{w,\beta}(\boldsymbol{\eta
}^{\ast})\left.  \frac{\partial V_{\beta}(\boldsymbol{\eta}_{1}%
,\boldsymbol{\eta}^{\ast})}{\partial\boldsymbol{\eta}_{1}}\right\vert
_{\boldsymbol{\eta}_{1}=\boldsymbol{\eta}^{\ast}}.
\end{align*}

\end{theorem}

\begin{proof}
The asymptotic distribution of $V_{\beta}(\widehat{\boldsymbol{\eta}}_{\beta
},\widehat{\boldsymbol{\eta}}_{\beta})$ coincides with the asymptotic
distribution of $V_{\beta}(\widehat{\boldsymbol{\eta}}_{\beta}%
,\boldsymbol{\eta}^{\ast})$ because $\widehat{\boldsymbol{\eta}}_{\beta
}\underset{n_{1},n_{2}\rightarrow\infty}{\overset{\mathcal{P}}{\longrightarrow
}}\boldsymbol{\eta}^{\ast}$. A Taylor expansion gives,%
\[
V_{\beta}(\widehat{\boldsymbol{\eta}}_{\beta},\widehat{\boldsymbol{\eta}%
}_{\beta})-V_{\beta}(\boldsymbol{\eta}^{\ast},\boldsymbol{\eta}^{\ast
})=\left.  \frac{\partial V_{\beta}(\boldsymbol{\eta}_{1},\boldsymbol{\eta
}^{\ast})}{\partial\boldsymbol{\eta}_{1}^{T}}\right\vert _{\boldsymbol{\eta
}_{1}=\boldsymbol{\eta}^{\ast}}(\widehat{\boldsymbol{\eta}}_{\beta
}-\boldsymbol{\eta}^{\ast})+o_{p}(\left\Vert \widehat{\boldsymbol{\eta}%
}_{\beta}-\boldsymbol{\eta}^{\ast}\right\Vert ).
\]
But we know
\[
\sqrt{\frac{n_{1}n_{2}}{n_{1}+n_{2}}}(\widehat{\boldsymbol{\eta}}_{\beta
}-\boldsymbol{\eta}^{\ast})\underset{n_{1},n_{2}\rightarrow\infty
}{\overset{\mathcal{L}}{\longrightarrow}}\mathcal{N}\left(  \boldsymbol{0}%
_{4},\boldsymbol{\Sigma}_{w,\beta}(\boldsymbol{\eta}^{\ast})\right)  ,
\]
therefore the desired result is obtained. 
\hfill\rule{0.7em}{0.7em}
\end{proof}

Based on the previous result we can establish an approximation of the power
function for the Wald-type test statistics given by (\ref{wald}). We have,%
\begin{align*}
\mathrm{Power}_{W_{n_{1},n_{2}}^{\beta}}\left(  \boldsymbol{\eta}^{\ast
}\right)   &  =P\left(  W_{n_{1},n_{2}}^{\beta}>\chi_{1,\alpha}%
^{2}|\boldsymbol{\eta}=\boldsymbol{\eta}^{\ast}\right)  \\
&  =P\left(  \tfrac{n_{1}n_{2}}{n_{1}+n_{2}}V_{\beta}%
(\widehat{\boldsymbol{\eta}}_{\beta},\widehat{\boldsymbol{\eta}}_{\beta}%
)>\chi_{1,\alpha}^{2}|\boldsymbol{\eta}=\boldsymbol{\eta}^{\ast}\right)  \\
&  =P\left(  \tfrac{n_{1}n_{2}}{n_{1}+n_{2}}\left(  V_{\beta}%
(\widehat{\boldsymbol{\eta}}_{\beta},\widehat{\boldsymbol{\eta}}_{\beta
})-V_{\beta}(\boldsymbol{\eta}^{\ast},\boldsymbol{\eta}^{\ast})\right)
>\chi_{1,\alpha}^{2}-\tfrac{n_{1}n_{2}}{n_{1}+n_{2}}V_{\beta}(\boldsymbol{\eta
}^{\ast},\boldsymbol{\eta}^{\ast})\right)  \\
&  =P\left(  V_{\beta}(\widehat{\boldsymbol{\eta}}_{\beta}%
,\widehat{\boldsymbol{\eta}}_{\beta})-V_{\beta}(\boldsymbol{\eta}^{\ast
},\boldsymbol{\eta}^{\ast})>\tfrac{n_{1}+n_{2}}{n_{1}n_{2}}\chi_{1,\alpha}%
^{2}-V_{\beta}(\boldsymbol{\eta}^{\ast},\boldsymbol{\eta}^{\ast})\right)  \\
&  =P\left(  \sqrt{\tfrac{n_{1}n_{2}}{n_{1}+n_{2}}}\left(  V_{\beta
}(\widehat{\boldsymbol{\eta}}_{\beta},\widehat{\boldsymbol{\eta}}_{\beta
})-V_{\beta}(\boldsymbol{\eta}^{\ast},\boldsymbol{\eta}^{\ast})\right)
>\sqrt{\tfrac{n_{1}n_{2}}{n_{1}+n_{2}}}\left(  \tfrac{n_{1}+n_{2}}{n_{1}n_{2}%
}\chi_{1,\alpha}^{2}-V_{\beta}(\boldsymbol{\eta}^{\ast},\boldsymbol{\eta
}^{\ast})\right)  \right)  \\
&  =P\left(  \frac{1}{\phi_{\beta,m}\left(  \boldsymbol{\eta}^{\ast}\right)
}\sqrt{\tfrac{n_{1}n_{2}}{n_{1}+n_{2}}}\left(  V_{\beta}%
(\widehat{\boldsymbol{\eta}}_{\beta},\widehat{\boldsymbol{\eta}}_{\beta
})-V_{\beta}(\boldsymbol{\eta}^{\ast},\boldsymbol{\eta}^{\ast})\right)
\right.  \\
&  \qquad\left.  >\frac{1}{\phi_{\beta,m}\left(  \boldsymbol{\eta}^{\ast
}\right)  }\sqrt{\tfrac{n_{1}n_{2}}{n_{1}+n_{2}}}\left(  \tfrac{n_{1}+n_{2}%
}{n_{1}n_{2}}\chi_{1,\alpha}^{2}-V_{\beta}(\boldsymbol{\eta}^{\ast
},\boldsymbol{\eta}^{\ast})\right)  \right)  \\
&  =1-\Phi_{n_{1},n_{2}}\left(  \frac{1}{\phi_{\beta,m}\left(
\boldsymbol{\eta}^{\ast}\right)  }\sqrt{\tfrac{n_{1}n_{2}}{n_{1}+n_{2}}%
}\left(  \tfrac{n_{1}+n_{2}}{n_{1}n_{2}}\chi_{1,\alpha}^{2}-V_{\beta
}(\boldsymbol{\eta}^{\ast},\boldsymbol{\eta}^{\ast})\right)  \right)  ,
\end{align*}
where $\Phi_{n_{1},n_{2}}(\cdot)$ is a sequence of distribution functions
tending uniformly to the standard normal  distribution function $\Phi\left(
\cdot\right)  $. It is clear that%

\[
\lim_{n_{1},n_{2}\rightarrow\infty}\mathrm{Power}_{W_{n_{1},n_{2}}^{\beta}%
}\left(  \boldsymbol{\eta}^{\ast}\right)  =1
\]
for all $\alpha\in\left(  0,1\right)  $. Therefore, 
Wald-type test statistics based on the MDPD estimator are consistent in the sense of  \cite{MR0083868}.

We may also find some other approximations of the power of $W_{n_{1},n_{2}}^{\beta}$ at
an alternative close to the null hypothesis. Let $\boldsymbol{\eta}%
_{n_{1},n_{2}}\in\Lambda-\Lambda_{0}$ be a given alternative and let
$\boldsymbol{\eta}_{0}$ be the element in $\Lambda_{0}$ closest to
$\boldsymbol{\eta}_{n_{1},n_{2}}$ in the Euclidean distance sense. A first
possibility to introduce contiguous alternative hypotheses is to consider a
fixed vector $\boldsymbol{d}\in\mathbb{R}^{4}$ and to permit $\boldsymbol{\eta
}_{n_{1},n_{2}}$ to move towards $\boldsymbol{\eta}_{0}$ as $n_{1}$ and $n_{2}$
increase in the  manner given in the hypothesis
\begin{equation}
H_{1,n_{1},n_{2}}:\boldsymbol{\eta}_{n_{1},n_{2}}=\boldsymbol{\eta}%
_{0}+\left(  \tfrac{n_{1}n_{2}}{n_{1}+n_{2}}\right)  ^{-1/2}\boldsymbol{d}%
.\label{4.3.1}%
\end{equation}
A second approach is to relax the condition $m(\boldsymbol{\eta})=0$ which
defines $\Theta_{0}$. Let $\delta\in\mathbb{R}$ and consider the following
sequence, $m(\boldsymbol{\eta}_{n_{1},n_{2}})$, of parameters moving towards
$m(\boldsymbol{\eta}_{0})$ according to
\begin{equation}
H_{1,n_{1},n_{2}}^{\ast}:m(\boldsymbol{\eta}_{n_{1},n_{2}})=\left(
\tfrac{n_{1}n_{2}}{n_{1}+n_{2}}\right)  ^{-1/2}\delta.\label{4.3.2}%
\end{equation}
Note that a Taylor series expansion of $m(\boldsymbol{\eta}_{n_{1},n_{2}})$
around $\boldsymbol{\eta}_{0}$ yields
\begin{equation}
m(\boldsymbol{\eta}_{n_{1},n_{2}})=m(\boldsymbol{\eta}_{0})+\left.
\frac{\partial m(\boldsymbol{\eta})}{\partial\boldsymbol{\eta}^{T}}\right\vert
_{\boldsymbol{\eta}=\boldsymbol{\eta}_{0}}\left(  \boldsymbol{\eta}%
_{n_{1},n_{2}}-\boldsymbol{\eta}_{0}\right)  +o\left(  \left\Vert
\boldsymbol{\eta}_{n_{1},n_{2}}-\boldsymbol{\eta}_{0}\right\Vert \right)
.\label{4.4}%
\end{equation}
By substituting $\boldsymbol{\eta}_{n_{1},n_{2}}=\boldsymbol{\eta}_{0}+\left(
\frac{n_{1}n_{2}}{n_{1}+n_{2}}\right)  ^{-1/2}\boldsymbol{d}$ in Equation (\ref{4.4})
and taking into account that\textbf{\ }$m(\boldsymbol{\eta}_{0})=0$, we get
\begin{equation}
m(\boldsymbol{\eta}_{n_{1},n_{2}})=\left(  \tfrac{n_{1}n_{2}}{n_{1}+n_{2}%
}\right)  ^{-1/2}\left.  \frac{\partial m(\boldsymbol{\eta})}{\partial
\boldsymbol{\eta}^{T}}\right\vert _{\boldsymbol{\eta}=\boldsymbol{\eta}_{0}%
}\boldsymbol{d}+o\left(  \left\Vert \boldsymbol{\eta}_{n_{1},n_{2}%
}-\boldsymbol{\eta}_{0}\right\Vert \right)  ,\label{4.4.b}%
\end{equation}
so that the equivalence in the limit is obtained for $\delta=\left.
\frac{\partial m(\boldsymbol{\eta})}{\partial\boldsymbol{\eta}^{T}}\right\vert
_{\boldsymbol{\eta}=\boldsymbol{\eta}_{0}}\boldsymbol{d}.$

We have the following result.

\begin{theorem}
We have
\begin{enumerate}
\item[i)] $W_{n_{1},n_{2}}^{\beta}\underset{n_{1},n_{2}\mathcal{\rightarrow
}\infty}{\overset{\mathcal{L}}{\longrightarrow}}\chi_{1}^{2}\left(
\frac{\delta^{2}}{\sigma_{\beta,m}^{2}(\boldsymbol{\eta}_{0})}\right)  $ under
$H_{1,n_{1},n_{2}}$ given in (\ref{4.3.1}).

\item[ii)] $W_{n_{1},n_{2}}^{\beta}\underset{n_{1},n_{2}\mathcal{\rightarrow
}\infty}{\overset{\mathcal{L}}{\longrightarrow}}\chi_{1}^{2}\left(
\frac{\left(  \left.  \frac{\partial m(\boldsymbol{\eta})}{\partial
\boldsymbol{\eta}^{T}}\right\vert _{\boldsymbol{\eta}=\boldsymbol{\eta}_{0}%
}\boldsymbol{d}\right)  ^{2}}{\sigma_{\beta,m}^{2}(\boldsymbol{\eta}_{0}%
)}\right)  $ under $H_{1,n_{1},n_{2}}^{\ast}$ given in (\ref{4.3.2}).
\end{enumerate}
\end{theorem}

\begin{proof}
A Taylor series expansion of $m(\widehat{\boldsymbol{\eta}}_{\beta})$ around
$\boldsymbol{\eta}_{n_{1},n_{2}}$ yields%
\[
m(\widehat{\boldsymbol{\eta}}_{\beta})=m(\boldsymbol{\eta}_{n_{1},n_{2}%
})+\left.  \frac{\partial m(\boldsymbol{\eta})}{\partial\boldsymbol{\eta}^{T}%
}\right\vert _{\boldsymbol{\eta}=\boldsymbol{\eta}_{n_{1},n_{2}}%
}(\widehat{\boldsymbol{\eta}}_{\beta}-\boldsymbol{\eta}_{n_{1},n_{2}%
})+o\left(  \left\Vert \widehat{\boldsymbol{\eta}}_{\beta}-\boldsymbol{\eta
}_{n_{1},n_{2}}\right\Vert \right)  .
\]
From (\ref{4.4.b}), we have
\begin{eqnarray*}
 m(\widehat{\boldsymbol{\eta}}_{\beta})&=&\left(  \tfrac{n_{1}n_{2}}{n_{1}+n_{2}%
}\right)  ^{-1/2}\left.  \frac{\partial m(\boldsymbol{\eta})}{\partial
\boldsymbol{\eta}^{T}}\right\vert _{\boldsymbol{\eta}=\boldsymbol{\eta}_{0}%
}\boldsymbol{d}+\left.  \frac{\partial m(\boldsymbol{\eta})}{\partial
\boldsymbol{\eta}^{T}}\right\vert _{\boldsymbol{\eta}=\boldsymbol{\eta}%
_{n_{1},n_{2}}}(\widehat{\boldsymbol{\eta}}_{\beta}-\boldsymbol{\eta}%
_{n_{1},n_{2}})\\
&&+o\left(  \left\Vert \widehat{\boldsymbol{\eta}}_{\beta
}-\boldsymbol{\eta}_{n_{1},n_{2}}\right\Vert \right)  +o\left(  \left\Vert
\boldsymbol{\eta}_{n_{1},n_{2}}-\boldsymbol{\eta}_{0}\right\Vert \right)  .
\end{eqnarray*}
As $\sqrt{\frac{n_{1}n_{2}}{n_{1}+n_{2}}}\left(  o\left(  \left\Vert
\widehat{\boldsymbol{\eta}}_{\beta}-\boldsymbol{\eta}_{n_{1},n_{2}}\right\Vert
\right)  +o\left(  \left\Vert \boldsymbol{\eta}_{n_{1},n_{2}}-\boldsymbol{\eta
}_{0}\right\Vert \right)  \right)  =o_{p}\left(  1\right)  $ and
(\ref{eqTh0}), we have
\[
\sqrt{\tfrac{n_{1}n_{2}}{n_{1}+n_{2}}}m(\widehat{\boldsymbol{\eta}}_{\beta
})\underset{n_{1},n_{2}\mathcal{\rightarrow}\infty}{\overset{\mathcal{L}%
}{\longrightarrow}}\mathcal{N}(\left.  \tfrac{\partial m(\boldsymbol{\eta}%
)}{\partial\boldsymbol{\eta}^{T}}\right\vert _{\boldsymbol{\eta}%
=\boldsymbol{\eta}_{0}}\boldsymbol{d},\sigma_{\beta,m}^{2}(\boldsymbol{\eta
}_{0})),
\]
which is equivalent to
\[
\sqrt{\tfrac{n_{1}n_{2}}{n_{1}+n_{2}}}m(\widehat{\boldsymbol{\eta}}_{\beta
})\underset{n_{1},n_{2}\mathcal{\rightarrow}\infty}{\overset{\mathcal{L}%
}{\longrightarrow}}\mathcal{N}(\delta,\sigma_{\beta,m}^{2}(\boldsymbol{\eta
}_{0})).
\]
Since $\widehat{\boldsymbol{\eta}}_{\beta}\underset{n_{1},n_{2}\rightarrow
\infty}{\overset{\mathcal{P}}{\longrightarrow}}\boldsymbol{\eta}_{0}$,
$\frac{\sigma_{\beta,m}(\widehat{\boldsymbol{\eta}}_{\beta})}{\sigma_{\beta
,m}(\boldsymbol{\eta}_{0})}\underset{n_{1},n_{2}\rightarrow\infty
}{\overset{\mathcal{P}}{\longrightarrow}}1$, and applying Slutsky's theorem
$\sqrt{\frac{n_{1}n_{2}}{n_{1}+n_{2}}}\tfrac{m(\widehat{\boldsymbol{\eta}%
}_{\beta})}{\sigma_{\beta,m}(\widehat{\boldsymbol{\eta}}_{\beta})}$ and \\
$
\sqrt{\tfrac{n_{1}n_{2}}{n_{1}+n_{2}}}\tfrac{\sigma_{\beta,m}%
(\widehat{\boldsymbol{\eta}}_{\beta})}{\sigma_{\beta,m}(\boldsymbol{\eta}%
_{0})}\tfrac{m(\widehat{\boldsymbol{\eta}}_{\beta})}{\sigma_{\beta
,m}(\widehat{\boldsymbol{\eta}}_{\beta})}
$
have the same asymptotic distribution. But,
\[
\sqrt{\tfrac{n_{1}n_{2}}{n_{1}+n_{2}}}\tfrac{\sigma_{\beta,m}%
(\widehat{\boldsymbol{\eta}}_{\beta})}{\sigma_{\beta,m}(\boldsymbol{\eta}%
_{0})}\tfrac{m(\widehat{\boldsymbol{\eta}}_{\beta})}{\sigma_{\beta
,m}(\widehat{\boldsymbol{\eta}}_{\beta})}\underset{n_{1},n_{2}%
\mathcal{\rightarrow}\infty}{\overset{\mathcal{L}}{\longrightarrow}%
}\mathcal{N}(\delta,1).
\]
Finally, the desired result is obtained from
\begin{equation}
W_{n_{1},n_{2}}^{\beta}=\left(  \sqrt{\tfrac{n_{1}n_{2}}{n_{1}+n_{2}}}%
\tfrac{m(\widehat{\boldsymbol{\eta}}_{\beta})}{\sigma_{\beta,m}%
(\widehat{\boldsymbol{\eta}}_{\beta})}\right)  ^{2}.
\label{wald_stat}
\end{equation}
\hfill\rule{0.7em}{0.7em}
\end{proof}

\section{Robustness of the Wald-type test statistics\label{sec4}}

In this Section we are going to get the partial influence functions in the
sense of \cite{MR1945963} for the Wald-type test statistics,
$W_{n_{1},n_{2}}^{\beta}(\widehat{\boldsymbol{\eta}}_{\beta})$, considered in
this paper.

Let $\boldsymbol{\eta}_{0}=\left(  \mathcal{\mu}_1^*,\sigma_1^*%
,\mathcal{\mu}_2^*,\sigma_2^*\right)  ^{T}\in\Lambda_{0}$. We shall
denote
\[
\boldsymbol{T}_{\beta}\left(  F_{\mu_1^*,\sigma_1^*},F_{\mu_2^*%
,\sigma_2^*}\right)  =\left(  \boldsymbol{T}_{\beta}^{1}\left(  F_{\mu
_1^*,\sigma_1^*}\right)  ,\boldsymbol{T}_{\beta}^{2}\left(  F_{\mu
_2^*,\sigma_2^*}\right)  \right)  ^{T}=\boldsymbol{\eta}_{0}\in\Lambda_{0}%
\]
where $F_{\mu_i^*,\sigma_i^*}$ the distribution function of a log-normal
population with the true values of parameters $\mu_{i}$ and $\sigma_{i},$
$i=1,2$ and $\boldsymbol{T}_{\beta}^{i}\left(  F_{\mu_i^*,\sigma_i^*%
}\right)  $ is the MDPD functional in log-normal
population with the true values of parameters $\mu_{i}$ and $\sigma_{i},$
$i=1,2.$ We shall also denote
\[
F_{\mu_i^*,\sigma_i^*}^{\varepsilon_i}=(1-\varepsilon_i)F_{\mu_i^*%
,\sigma_i^*}+\varepsilon_i\delta_{x},\text{ }i=1,2
\]
with $\varepsilon_i\in\left(  0,1\right)  $ and $\delta_{x}$ is a point mass
distribution at $x.$

The functional associated with the Wald-type test statistics $W_{n_{1},n_{2}%
}^{\beta}(\widehat{\boldsymbol{\eta}}_{\beta})$, without the constant factor $\frac{n_1 n_2}{n_1+n_2}$, is given by%
\begin{equation}
w_{n_{1},n_{2}}^{\beta}\left(  \boldsymbol{T}_{\beta}\left(  F_{\mu
_1^*,\sigma_1^*},F_{\mu_2^*,\sigma_2^*}\right)  \right)  =\frac
{m^{2}\left(  \boldsymbol{T}_{\beta}\left(  F_{\mu_1^*,\sigma_1^*}%
,F_{\mu_2^*,\sigma_2^*}\right)  \right)  }{\sigma_{\beta,m}^{2}%
(\boldsymbol{T}_{\beta}(F_{\mu_1^*,\sigma_1^*},F_{\mu_2^*,\sigma_2^*%
}))}.
\label{wald_functional}
\end{equation}
In accordance with \cite{MR1945963} the partial influence functions are
given by
\[
\mathcal{IF}\left(  x, w_{n_{1},n_{2}}^{\beta},F_{\mu_1^*,\sigma_1^*%
}|F_{\mu_2^*,\sigma_2^*}\right)  =\left.  \frac{\partial w_{n_{1},n_{2}%
}^{\beta}\left(  F_{\mu_1^*,\sigma_1^*}^{\varepsilon_1},F_{\mu_2^*%
,\sigma_2^*}\right)  }{\partial\varepsilon_1}\right\vert _{\varepsilon_1=0}%
\]
and
\[
\mathcal{IF}\left(  x, w_{n_{1},n_{2}}^{\beta},F_{\mu_2^*,\sigma_2^*%
}|F_{\mu_1^*,\sigma_1^*}\right)  =\left.  \frac{\partial w_{n_{1},n_{2}%
}^{\beta}\left(  F_{\mu_1^*,\sigma_1^*},F_{\mu_2^*,\sigma_2^*%
}^{\varepsilon_2}\right)  }{\partial\varepsilon_2}\right\vert _{\varepsilon_2=0}.
\]
By analogy with the one-sample case,  $\mathcal{IF}\left(  x, w_{n_{1},n_{2}}^{\beta},F_{\mu_1^*,\sigma_1^*%
}|F_{\mu_2^*,\sigma_2^*}\right)$
measures, approximately, $n_1$ times the change on  $w_{n_{1},n_{2}}^{\beta}$ caused by an additional observation in $x$, 
when it is applied to a large combined sample of $(n_1, n_2)$ observations. Similarly, the second partial influence function is interpreted. 
Now, for the Wald-type test statistic, we have
\begin{align*}
&  \mathcal{IF}\left(  x,w_{n_{1},n_{2}}^{\beta},F_{\mu_1^*,\sigma_1^*%
}|F_{\mu_2^*,\sigma_2^*}\right)  \\
&  =\Bigg(  \left.  \frac{\partial m(\boldsymbol{\eta})}{\partial
\boldsymbol{\eta}^{T}}\right\vert _{\boldsymbol{\eta}=\boldsymbol{T}_{\beta
}\left(  F_{\mu_1^*,\sigma_1^*},F_{\mu_2^*,\sigma_2^*}\right)  }%
\\
& -\frac{m\left(  \boldsymbol{T}_{\beta}\left(  F_{\mu_1^*,\sigma_1^*%
},F_{\mu_2^*,\sigma_2^*}\right)  \right)  }{\sigma_{\beta,m}%
(\boldsymbol{T}_{\beta}(F_{\mu_1^*,\sigma_1^*},F_{\mu_2^*,\sigma_2^*%
}))}\left.  \frac{\partial\sigma_{\beta,m}^{2}(\boldsymbol{\eta})}%
{\partial\boldsymbol{\eta}^{T}}\right\vert _{\boldsymbol{\eta}=\boldsymbol{T}%
_{\beta}\left(  F_{\mu_1^*,\sigma_1^*},F_{\mu_2^*,\sigma_2^*}\right)
}\Bigg)  \\
&  \times2\frac{m\left(  \boldsymbol{T}_{\beta}\left(  F_{\mu_1^*%
,\sigma_1^*},F_{\mu_2^*,\sigma_2^*}\right)  \right)  }{\sigma_{\beta
,m}^{2}(\boldsymbol{T}_{\beta}(F_{\mu_1^*,\sigma_1^*},F_{\mu_2^*%
,\sigma_2^*}))}\mathcal{IF}\left(  x,\boldsymbol{T}_{\beta}^{1},F_{\mu
_1^*,\sigma_1^*}\right) \\ 
& =0.
\end{align*}
The last equality comes from $m\left(  \boldsymbol{T}_{\beta}\left(
F_{\mu_1^*,\sigma_1^*},F_{\mu_2^*,\sigma_2^*}\right)  \right)  =0$. In
a similar manner,%
\[
\mathcal{IF}\left(  x, w_{n_{1},n_{2}}^{\beta},F_{\mu_2^*,\sigma_2^*%
}|F_{\mu_1^*,\sigma_1^*}\right)  =0.
\]
Thus, the first order partial influence function is not useful in quantifying
the robustness of these tests. Now we are going to get the second order
partial influence functions for the Wald-type test statistics.  It is a simple
exercise to see that
\begin{align*}
&  \mathcal{IF}_{2}\left(  x, w_{n_{1},n_{2}}^{\beta},F_{\mu_1^*,\sigma
_1^*}|F_{\mu_2^*,\sigma_2^*}\right)  =\left.  \frac{\partial^{2}%
w_{n_{1},n_{2}}^{\beta}\left(  F_{\mu_1^*,\sigma_1^*}^{\varepsilon_1}%
,F_{\mu_2^*,\sigma_2^*}\right)  }{\partial\varepsilon_1^{2}}\right\vert
_{\varepsilon_1=0}\\
&  =\frac{2}{\sigma_{\beta,m}^{2}(\boldsymbol{T}_{\beta}(F_{\mu_1^*%
,\sigma_1^*},F_{\mu_2^*,\sigma_2^*}))}\left(  \left.  \frac{\partial
m(\boldsymbol{\eta})}{\partial\boldsymbol{\eta}^{T}}\right\vert
_{\boldsymbol{\eta}=\boldsymbol{T}_{\beta}\left(  F_{\mu_1^*,\sigma_1^*%
},F_{\mu_2^*,\sigma_2^*}\right)  }\mathcal{IF}\left(  x,\boldsymbol{T}%
_{\beta}^{1},F_{\mu_1^*,\sigma_1^*}\right)  \right)  ^{2}.
\end{align*}
Similarly,
\begin{align*}
&  \mathcal{IF}_{2}\left(  x,w_{n_{1},n_{2}}^{\beta},F_{\mu_2^*,\sigma
_2^*}|F_{\mu_1^*,\sigma_1^*}\right)  =\left.  \frac{\partial^{2}%
w_{n_{1},n_{2}}^{\beta}\left(  F_{\mu_1^*,\sigma_1^*},F_{\mu_2^*%
,\sigma_2^*}^{\varepsilon_2}\right)  }{\partial\varepsilon_2^{2}}\right\vert
_{\varepsilon_2=0}\\
&  =\frac{2}{\sigma_{\beta,m}^{2}(\boldsymbol{T}_{\beta}(F_{\mu_1^*%
,\sigma_1^*},F_{\mu_2^*,\sigma_2^*}))}\left(  \left.  \frac{\partial
m(\boldsymbol{\eta})}{\partial\boldsymbol{\eta}^{T}}\right\vert
_{\boldsymbol{\eta}=\boldsymbol{T}_{\beta}\left(  F_{\mu_1^*,\sigma_1^*%
},F_{\mu_2^*,\sigma_2^*}\right)  }\mathcal{IF}\left(  x,\boldsymbol{T}%
_{\beta}^{2},F_{\mu_2^*,\sigma_2^*}\right)  \right)  ^{2},
\end{align*}
and
\begin{align*}
&  \mathcal{IF}_{2}\left(  x, w_{n_{1},n_{2}}^{\beta},F_{\mu_1^*,\sigma
_1^*}, F_{\mu_2^*,\sigma_2^*}\right)  =\left.  \frac{\partial^{2}%
w_{n_{1},n_{2}}^{\beta}\left(  F_{\mu_1^*,\sigma_1^*}^{\varepsilon_1}%
,F^{\varepsilon_2}_{\mu_2^*,\sigma_2^*}\right)  }{\partial\varepsilon_1 \partial\varepsilon_2}\right\vert
_{\varepsilon_1=0, \varepsilon_2=0}\\
&  =\frac{2 \mathcal{IF}\left(  x,\boldsymbol{T}%
_{\beta}^{1},F_{\mu_1^*,\sigma_1^*}\right) \mathcal{IF}\left(  x,\boldsymbol{T}%
_{\beta}^{2},F_{\mu_2^*,\sigma_2^*}\right)}{\sigma_{\beta,m}^{2}(\boldsymbol{T}_{\beta}(F_{\mu_1^*%
,\sigma_1^*},F_{\mu_2^*,\sigma_2^*}))}\left(  \left.  \frac{\partial
m(\boldsymbol{\eta})}{\partial\boldsymbol{\eta}^{T}}\right\vert
_{\boldsymbol{\eta}=\boldsymbol{T}_{\beta}\left(  F_{\mu_1^*,\sigma_1^*%
},F_{\mu_2^*,\sigma_2^*}\right)  }\right)  ^{2}  .
\end{align*}

As the first order partial influence functions vanishes, using a Taylor series expansion it can be shown that the bias in the functional of the Wald-type statistic is approximated by
\begin{align}
& w_{n_{1},n_{2}}^{\beta}\left(  \boldsymbol{T}_{\beta}\left(  F^{\varepsilon_1}_{\mu
_1^*,\sigma_1^*},F^{\varepsilon_2}_{\mu_2^*,\sigma_2^*}\right)  \right) - w_{n_{1},n_{2}}^{\beta}\left(  \boldsymbol{T}_{\beta}\left(  F_{\mu
_1^*,\sigma_1^*},F_{\mu_2^*,\sigma_2^*}\right)  \right)\nonumber\\
& \approx \frac{\varepsilon_1^2}{2} \mathcal{IF}_{2}\left(  x, w_{n_{1}, n_{2}}^{\beta},F_{\mu_1^*,\sigma
_1^*}|F_{\mu_2^*,\sigma_2^*}\right) + \varepsilon_1 \varepsilon_2 \mathcal{IF}_{2}\left(  x, w_{n_{1},n_{2}}^{\beta},F_{\mu_1^*,\sigma
_1^*}, F_{\mu_2^*,\sigma_2^*}\right)\nonumber\\
& \ \ \ \ \  + \frac{\varepsilon_2^2}{2}  \mathcal{IF}_{2}\left(  x, w_{n_{1},n_{2}}^{\beta},F_{\mu_2^*,\sigma
_2^*}|F_{\mu_1^*,\sigma_1^*}\right) .
\label{bias}
\end{align}
Note that in Equation (\ref{wald_functional}), we have not included the factor $\frac{n_1 n_2}{n_1+n_2}$, so the approximate bias in the Wald-type statistic in is $\frac{n_1 n_2}{n_1+n_2}$ times of the above mentioned bias given in Equation (\ref{bias}).
In Section \ref{sec2}, for the log-normal distribution, we proved that $\mathcal{IF}\left(  x,\boldsymbol{T}%
_{\beta}^{i},F_{\mu_1^*,\sigma_1^*}\right)$, $ i=1,2$, the influence function of the power divergence estimator,  is bounded when $\beta>0$, while it is unbounded when $\beta=0$. So, all three second order partial influential functions used in calculating the bias term in Equation (\ref{bias}) are bounded for $\beta>0$. Hence, the Wald-type tests are robust against outliers. It will be further verified by the simulation results.

\begin{figure}
\centering 
\begin{tabular}
[c]{cc}%
\raisebox{-0cm}{\includegraphics[height=9.217cm,width=7.5cm]{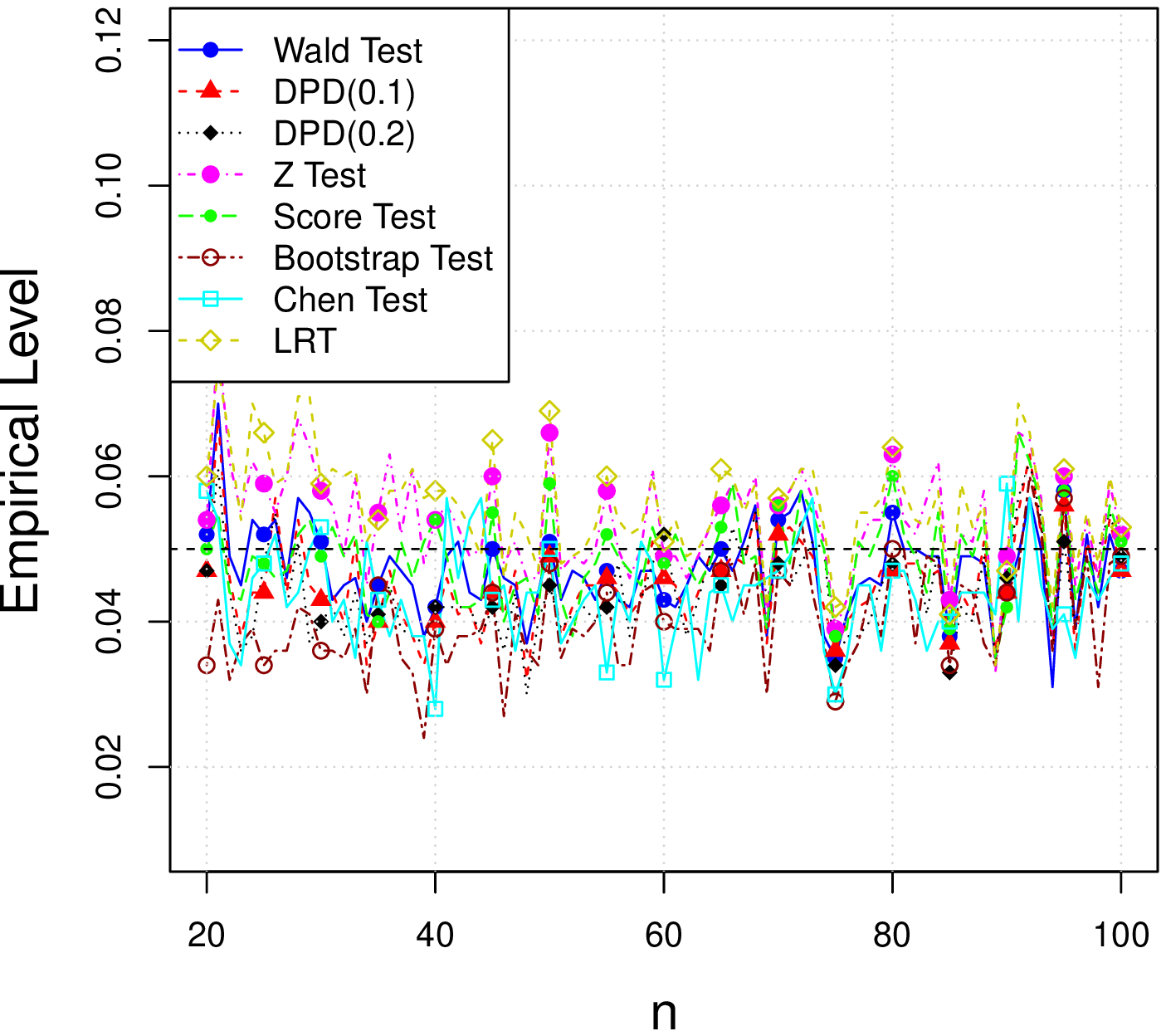}}
& 
{\includegraphics[height=9.217cm,width=7.5cm]{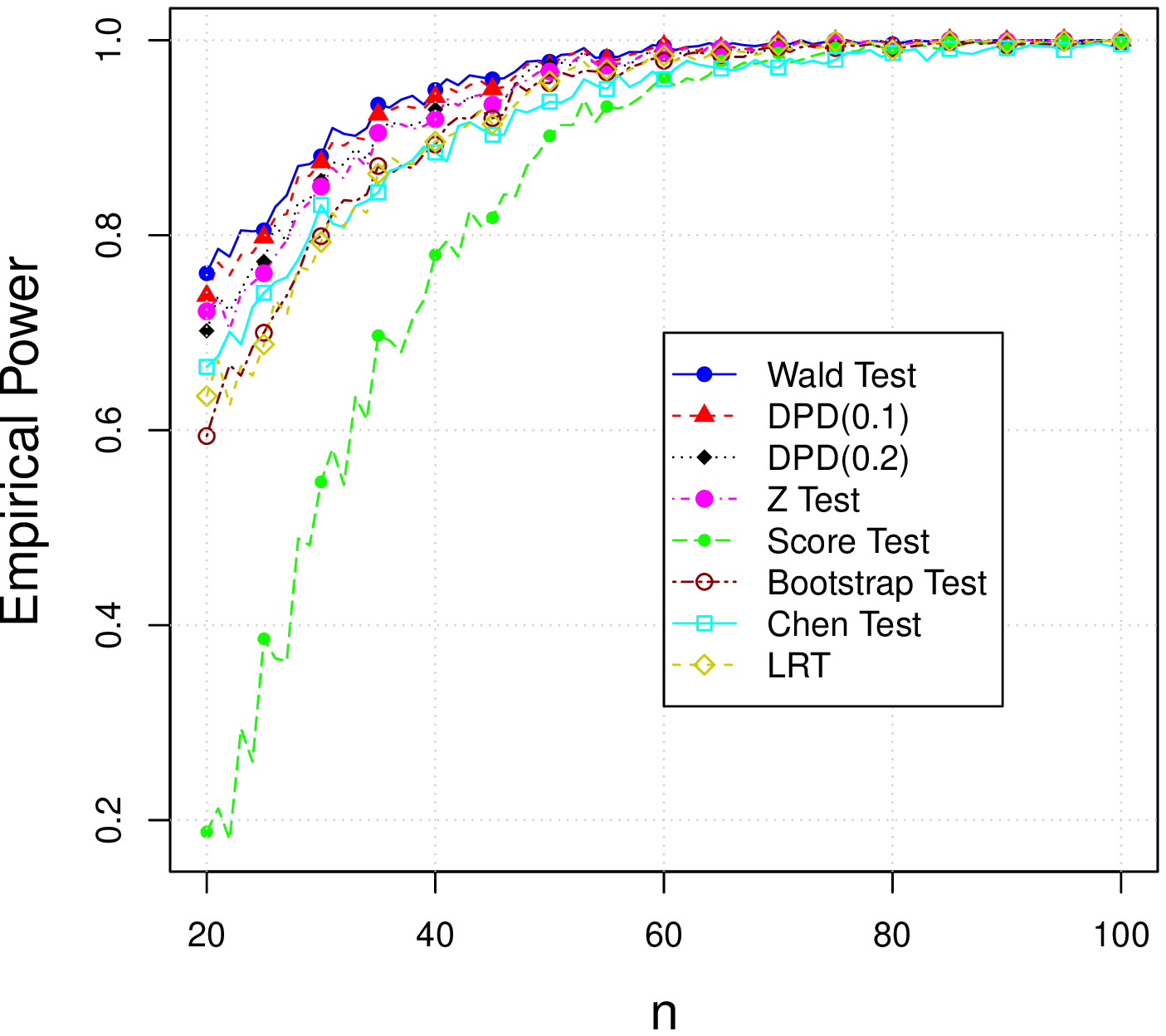}}\\
{\small (a)} & {\small (b)}\\%
\raisebox{-0cm}{\includegraphics[height=9.217cm,width=7.5cm]{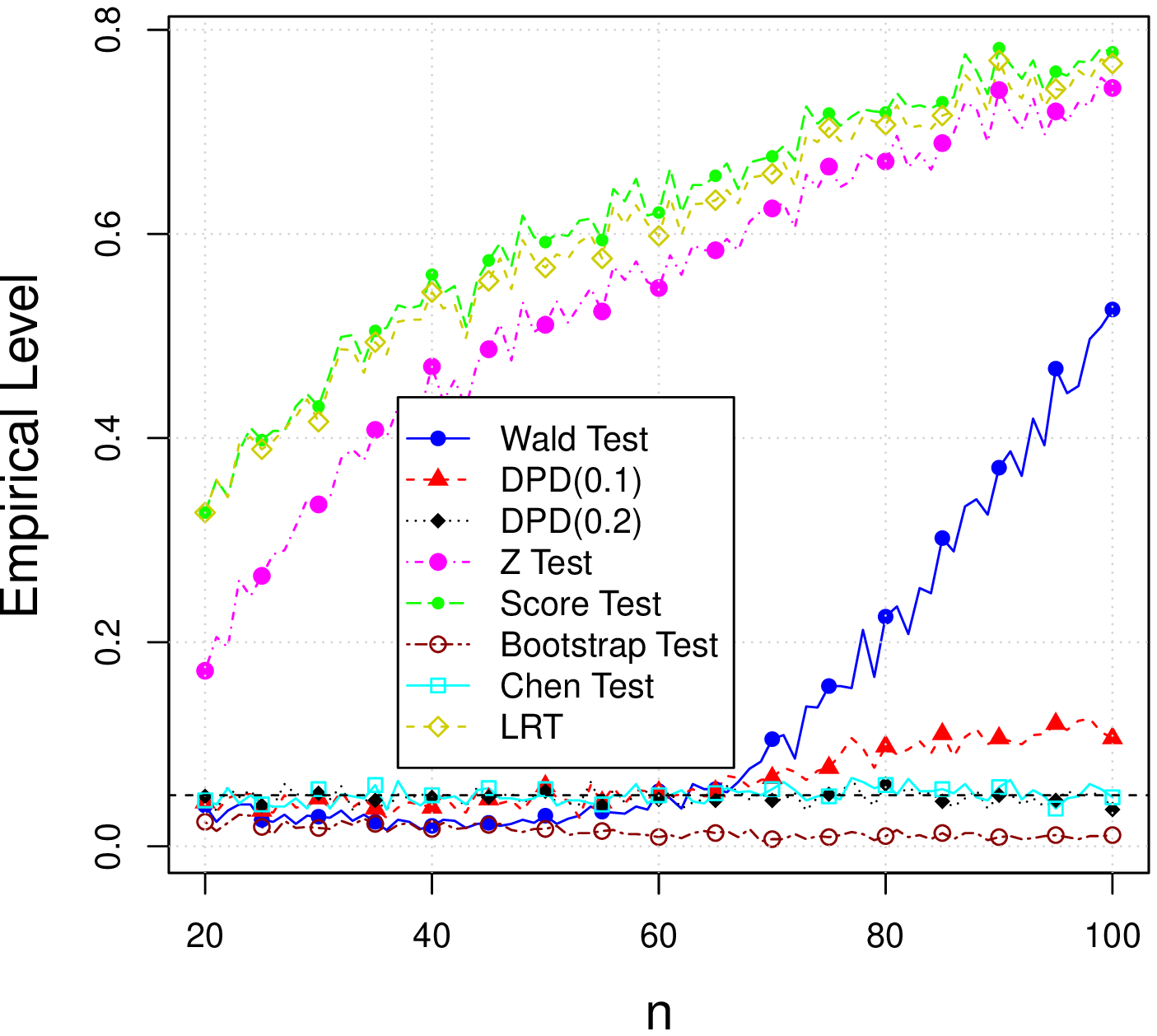}}
& 
\raisebox{-0cm}{\includegraphics[height=9.217cm,width=7.5cm]{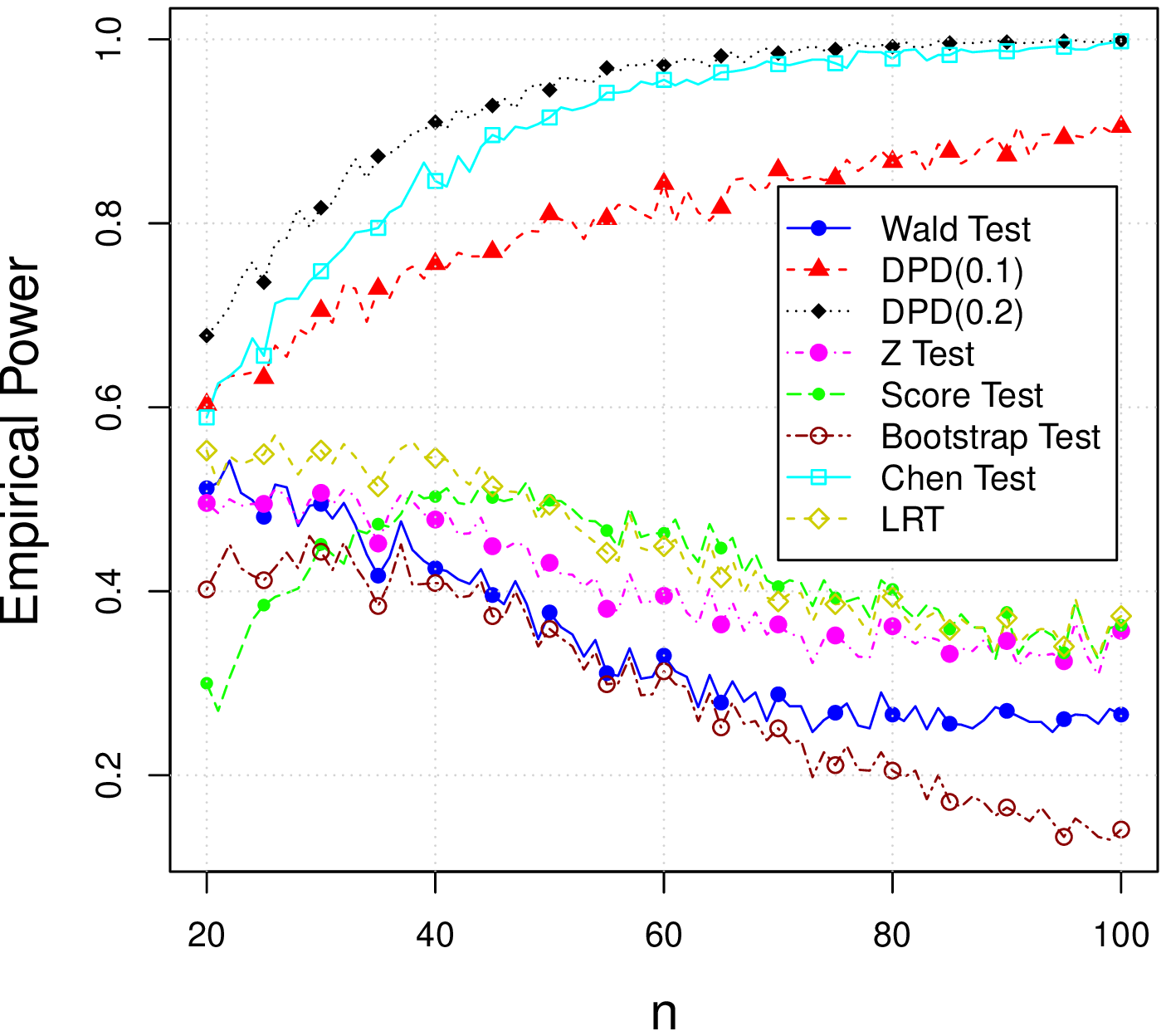}}\\
{\small (c)} & {\small (d)}%
\end{tabular}
\caption{Power and level under equal variances: (a) empirical level under pure data, (b) empirical power under pure
data, (c) empirical level under contaminated data and (d) empirical power
under contaminated data. }
\label{fig:raw}%
\end{figure}%

\section{Simulation study\label{sec6}}%

In this section we have studied the performance of the proposed Wald-type tests  using simulated data. We are interested in testing the null hypothesis given in (\ref{A}). We have divided the simulation setup in two parts -- one with  equal variances for the two populations, and  other with  unequal variances. For each case we have presented four plots -- level and power under pure data as well as under contaminated data. All tests are performed at the 5\% nominal level. At the end of the simulation results, we have presented a guideline to choose the tuning parameter $\beta$.

In the first case the values of the parameters are taken as $\mu_{1}=\mu_{2}=0$ and $\sigma_{1}^{2}=\sigma_{2}^{2}=0.4$.  
The samples from the two populations are chosen such that the ratio of the sample sizes (the first to the second) remains fixed at 1.5. This ratio is held constant as the overall sample sizes increase indefinitely.  Let $n$ denote the overall sample size of the two samples combined. The  observed levels are presented in Figure \ref{fig:raw}(a) for different values of $n$. Here the  observed level  is measured as the proportion of test statistics exceeding the corresponding critical value in 1000 replications. We have taken
 two  Wald-type test statistics based on MDPD estimator for $\beta= 0.1$ and $0.2$, denoted by $DPD(\beta)$. For comparison we have taken the classical Wald test,   $Z$ test and the bootstrap test proposed by \cite{zhou1997methods}, the score test proposed by \cite{gupta2006statistical}, Chen test proposed by \cite{MR1204372} and the likelihood ratio test (LRT). Note that 
 the classical Wald test statistic is a special case of the proposed family of the Wald-type test  statistics corresponding to $\beta=0$. In the bootstrap test we have taken 500 re-samples.  Figure \ref{fig:raw}(a) shows that the observed levels
 of all these tests are very close to the nominal level of 0.05 even at fairly small sample sizes. 
 
In order to get the powers of different tests, we simulate data from lognormal distributions where $\mu_1 = 0.8$ but the other parameters are identical to the level exploration case described in the previous paragraph. The null hypothesis is now an incorrect one and  observed powers (obtained in a similar manner as above) of the tests are presented in Figure \ref{fig:raw}(b). It is noticed that the classical Wald test performs best in terms of the power of the test, and the power of the score test is very small at least in  small sample sizes.  However, this discrepancy decreases
rapidly with the sample size, and by the time $n\geq 100$, all the observed
powers are practically equal to one. In any case for $\beta=0.1$ and 0.2 the Wald-type tests are almost as powerful as the classical Wald test and appear to be superior to all the other competitors considered.  

One point is worth noting here. The levels and powers in the graphs of Figures \ref{fig:raw}(a) and \ref{fig:raw}(b) have been calculated under $\sigma_1 = \sigma_2$, but without an assumption of the equality of these two parameters in the hypothesis. The only test for which we have made this assumption is the Chen test, which does not work otherwise. In that respect, the Chen test presented here in Figures \ref{fig:raw}(a) and \ref{fig:raw}(b) is not strictly comparable with the others in these graphs. But the point that is to be noted here is that the Wald-type tests with $\beta = 0.1$ and 0.2 are competitive or better than the Chen test in these cases, even though the Chen test enjoys more information and more structure in its construction. 

To evaluate the stability of the level and the power of the tests under
contamination, we repeated the above two sets of  simulations using contaminated data. We have replaced 5\% of the observations in the second population with those  from a log-normal distribution with $\mu=5$ and $\sigma^2=0.4$ while the population 1 model is the same as the one considered in Figure \ref{fig:raw}(a).  Figure \ref{fig:raw}(c) gives the empirical levels of the different tests in this case. 
It is interesting to note that only the tests except the Wald-type test (with $\beta = 0.2$) and the Chen test, perform very well in terms of the stability of the level. The bootstrap test is extremely conservative. The Wald-type test ($\beta = 0.1$) performs slightly worse than the best ones but is still relatively stable. All the others are hopelessly inadequate, and their levels eventually shoot up to 1 (the plot is given only for $n \leq 100)$. 

In the above set of simulations, apart from the Wald-type tests with moderately large positive values of $\beta$, the Chen test also appears to be reasonably stable. Yet at some other contaminations or larger sample sizes the Chen test could eventually fail to produce desired results. This is confirmed in the plot of Figure \ref{chen}, where the simulation is with a different contamination.  Beyond a sample size of 70 or so the Chen test exhibits an increasing trend, and fairly soon becomes irrelevant.

Finally, we have plotted the power functions of the tests for the contaminated data in Figure \ref{fig:raw}(d). The set up here is the same as the set up of Figure \ref{fig:raw}(c), expect that $\mu_1$ equals 0.8.   The powers of the robust Wald-type test with $\beta  = 0.2$ and Chen test show hardly any change under contamination. The $DPD(0.1)$ test is moderately stable, but has a larger loss in power compared to the previous two. All the other five tests are seen to break down completely. 

In the next part of the simulation we allow the variances to be unequal and choose  $\sigma_1^2=0.4$ and $\sigma_2^2=0.2$. The plots are given in Figure \ref{fig:raw1}. The values of mean parameters in Figures \ref{fig:raw1}(a) and \ref{fig:raw1}(c)  are $\mu_{1}=1.1$, and $\mu_{2}=1.2$. Note that the means of the two populations are equal to $\exp(1.3) = 3.67$; i.e. the null hypothesis is true. In Figures \ref{fig:raw1}(b) and \ref{fig:raw1}(d) the means are $\mu_{1}=1.6$, and $\mu_{2}=1.2$, so the null hypothesis is false. In case of Figures \ref{fig:raw1}(c) and \ref{fig:raw1}(d) we have replaced  5\% of the observations in the second population with those from a log-normal distribution with $\mu=5$ and $\sigma^2=0.2$. The Chen test is not applicable in this situation as it requires equal variances. In fact it completely fails to maintain the level and power even in pure data. Accordingly, we have left the Chen test out of the simulations in Figure \ref{fig:raw1}. Apart from this, the results depicted in Figure \ref{fig:raw1} are generally similar to those of Figure \ref{fig:raw} except for minor reversals here and there.  Anyhow, it is clear that the proposed Wald-type tests with moderate values of $\beta$ outperforms the Chen test even where it is applicable. 
As a whole, for  contaminated data, the proposed Wald-type test with a moderate value of $\beta$ does significantly better than  other tests in preserving its level as well as power. On the other hand, the Wald-type tests are very competitive to them under pure data.

\begin{figure}
{\begin{center}
\includegraphics[
height=6cm,
width=10cm
]{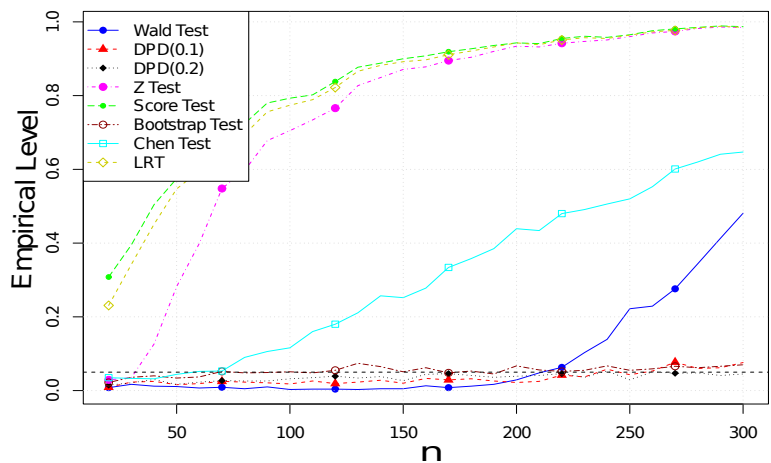}
\end{center}}
\caption{Plot of empirical levels of different tests when $\mu_{1}=\mu_{2}=0$ and $\sigma_{1}^{2}=\sigma_{2}^{2}=1$, and the second population is contaminated to the extent of 5\% with observations from a  log-normal distribution with $\mu=-10$ and $\sigma^2=1$.}
\label{chen}
\end{figure}

\begin{figure}
\centering 
\begin{tabular}
[c]{cc}%
\raisebox{-0cm}{\includegraphics[height=9.217cm,width=7.5cm]{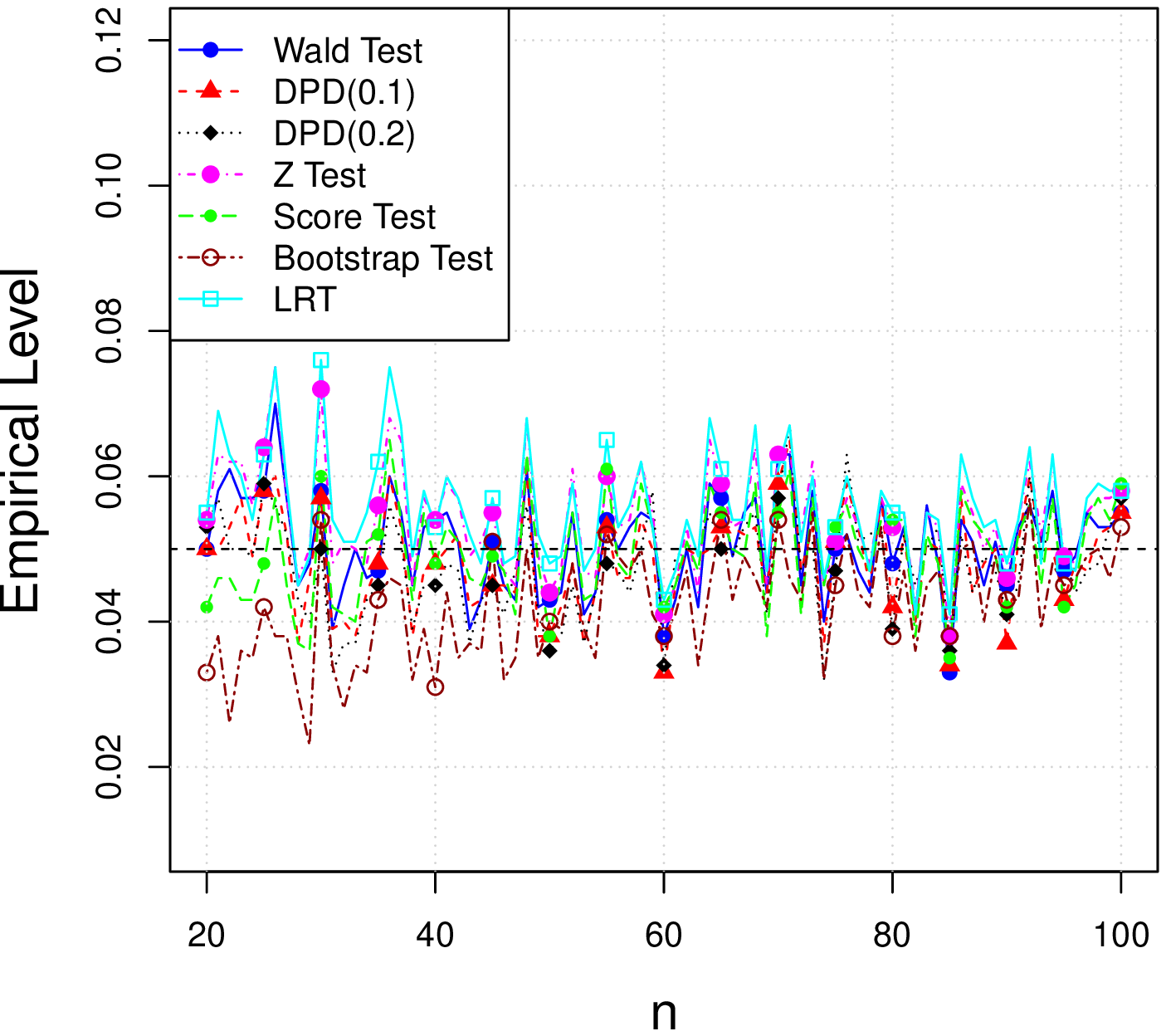}}
& 
{\includegraphics[height=9.217cm,width=7.5cm]{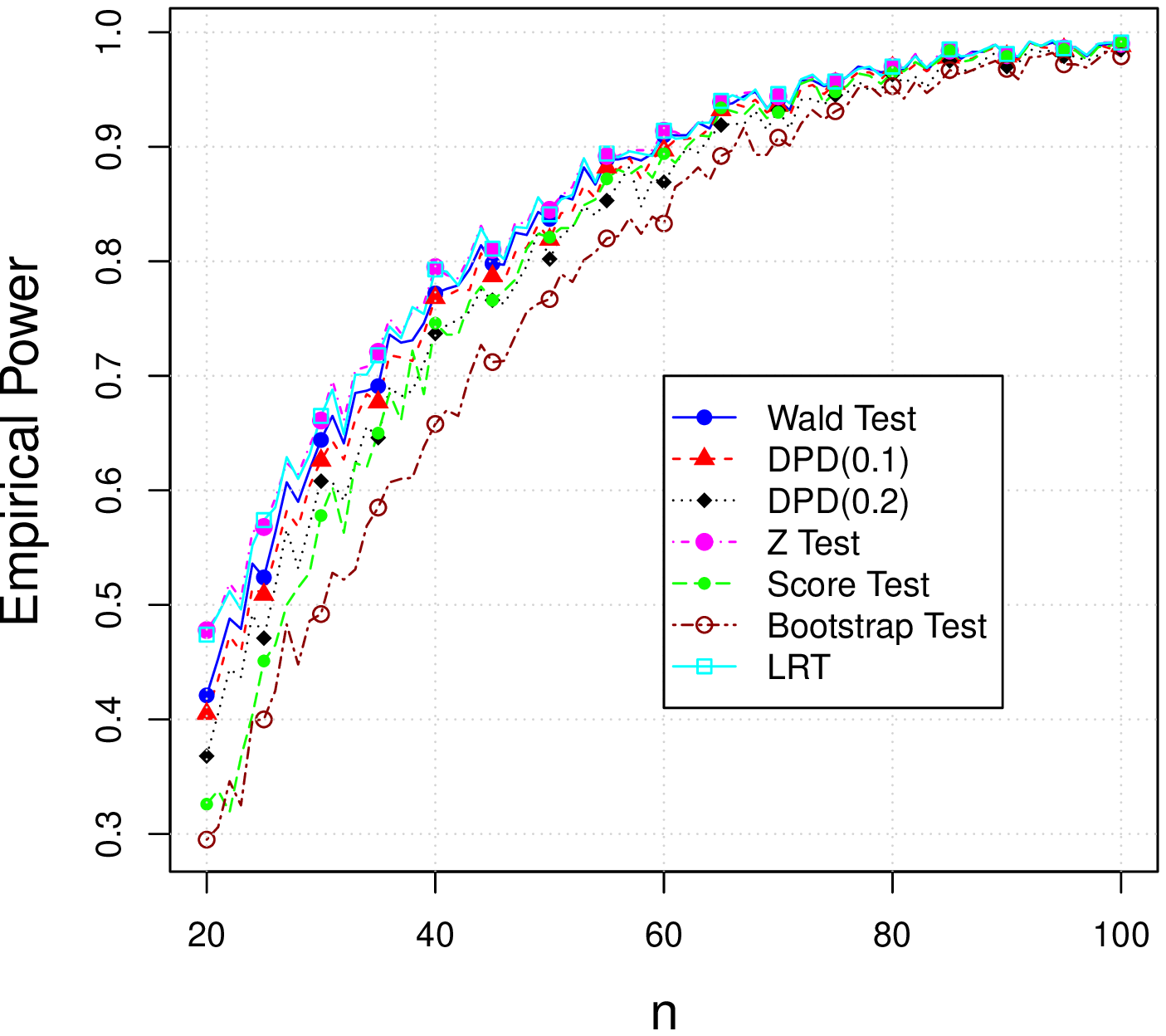}}\\
{\small (a)} & {\small (b)}\\%
\raisebox{-0cm}{\includegraphics[height=9.217cm,width=7.5cm]{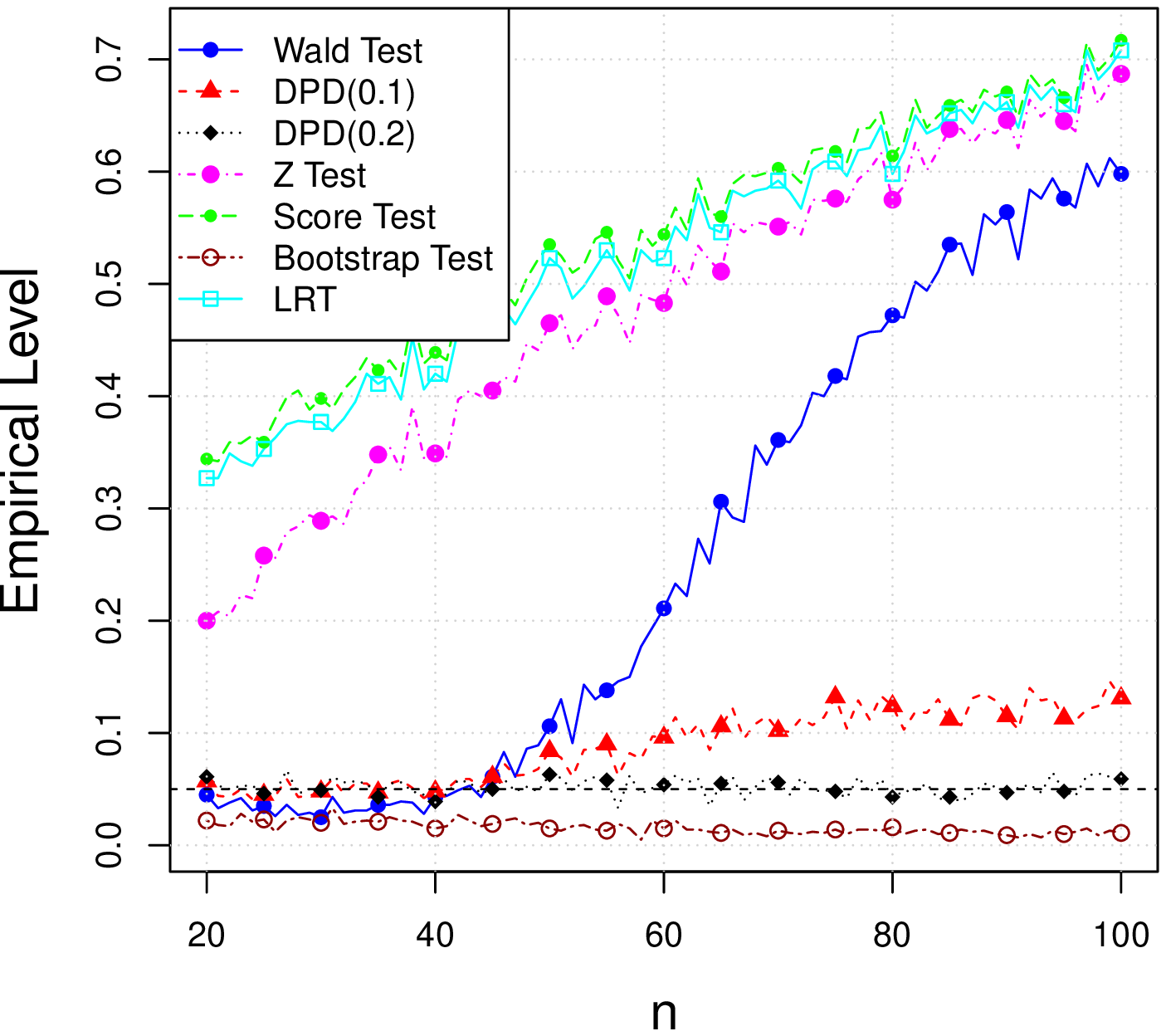}}
& 
\raisebox{-0cm}{\includegraphics[height=9.217cm,width=7.5cm]{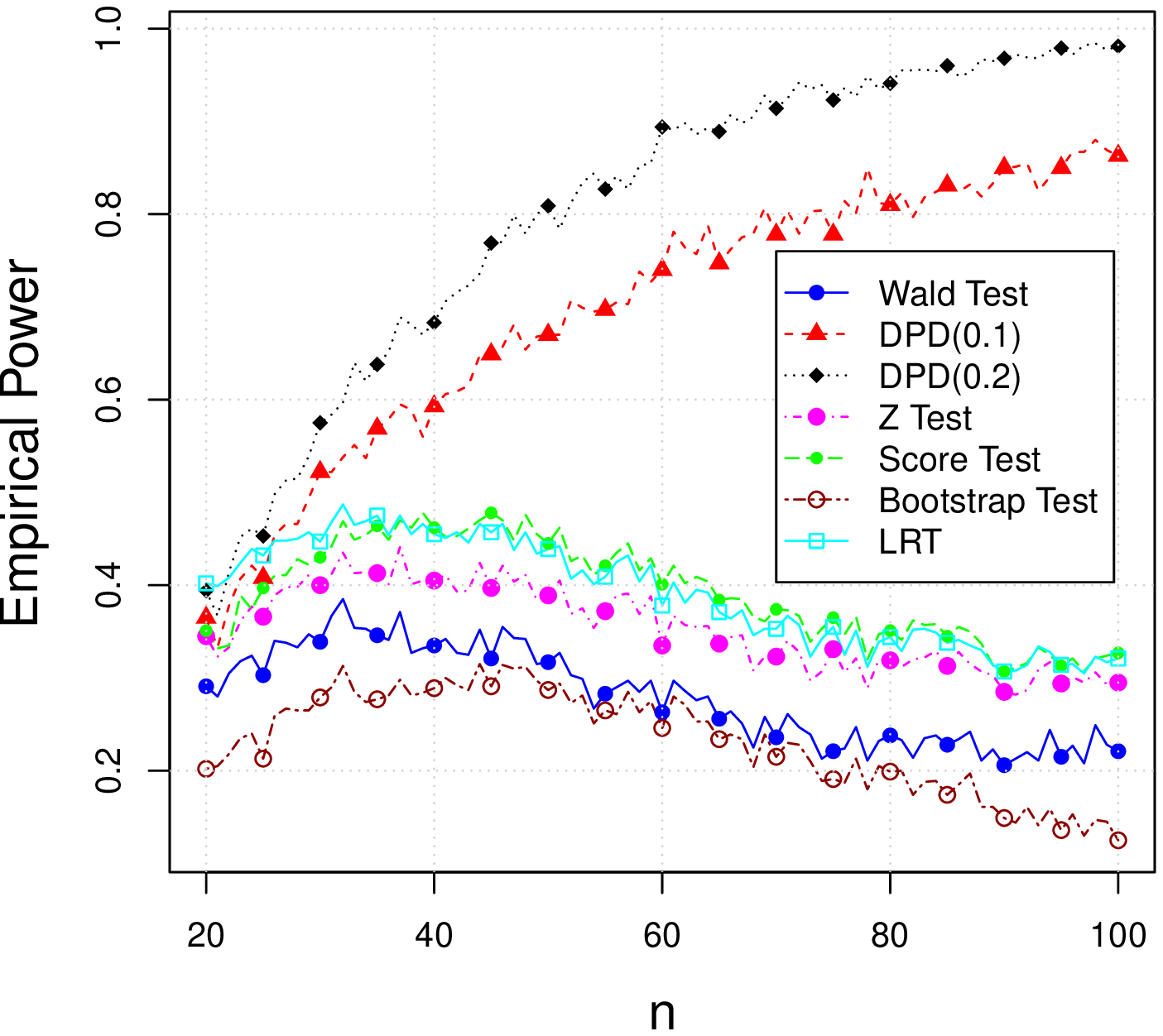}}\\
{\small (c)} & {\small (d)}%
\end{tabular}
\caption{Power and level under unequal variances: (a) empirical level under pure data, (b) empirical power under pure
data, (c) empirical level under contaminated data and (d) empirical power
under contaminated data. }
\label{fig:raw1}%
\end{figure}%

By construction, the Wald-type test statistic depends on the tuning parameter $\beta$. It is used to evaluate a robust estimator, and as a consequence, we obtain a robust test. The theoretical as well as simulation results show that the robustness properties of the Wald-type tests increase as $\beta$ increases. However, from Equation (\ref{as}), it can be shown that the efficiency of the MDPD estimators decrease as $\beta$ increases. Thus, for the large values of $\beta$, we observe a  loss of power in the Wald-type tests in the pure data, although they are robust in the contaminated data. Therefore, $\beta$ makes a trade-off between the efficiency and robustness of the MDPD estimator and the corresponding Wald-type test. In general, the robustness of the proposed tests correspond almost exactly to the robustness of the MDPD estimators, so we feel that the ``optimal'' choice of $\beta$ in the context of estimation would  work reasonably well in case of the hypothesis testing problem also. Now, the MLE, which is the MDPD estimator using  $\beta=0$, is the most efficient estimator. On the other hand, the density power divergence with $\beta=1$  corresponds to the $L_2$ distance that produces a strong robust estimator. So, we restrict the range of $\beta$ in the interval $[0,1]$. Values of $\beta \in [0.1, 0.2]$ are often reasonable choices, although tentative outliers and heavier contamination may require greater downweighting through a larger value of $\beta$.  Apart from the fixed choices of $\beta$, the data driven and adaptive choices could also be useful, as one can then tune the parameter to make the procedure more robust as required. One may follow the approach of \cite{Warwick} for this purpose, which minimizes an empirical measure of the mean square error of the estimator to determine the ``optimal'' tuning parameter. This requires the use of a robust pilot estimator of the parameter, and then the ``optimal'' estimator is obtained using an iterative method. \cite{Warwick} suggested the use of the MPDPD estimator corresponding to $\beta = 1$ as a pilot estimator.

\section{Numerical examples\label{sec5}}%
{\bf Example 1 (Air Quality data):} 
\cite{krishnamoorthy2003inferences}  analyzed a data set on air quality level from the {\it Data and Story Library} (\href{http://lib.stat.cmu.edu/DASL/}{http://lib.stat.cmu.edu/DASL/}). The data set measures the air quality from an oil refinery located at the
northeast of San Francisco. The refinery  conducted a series of 31 daily measurements of the carbon
monoxide levels arising from one of their stacks between April 16 and May 16, 1993.
It submitted these data to the Bay Area Air Quality Management District (BAAQMD)  as evidence for establishing a baseline for the air quality. BAAQMD personnel also made 9
independent measurements of the carbon monoxide concentration from the same stack over the period from September 11, 1990 to March 30, 1993. The data are given in Table \ref{tab:1}. We have checked the assumption of log-normality, and found that a  log-normal model adequately describes both sets of measurements but the normal model does not fit the data at any practical levels of significance. Under the log-normal model we want to test the null hypothesis that the means of two populations are equal. 

\begin{table}
 \begin{tabular}{ l  l }
\hline
 Refinery &  21,  30,  30,  34,  36,  37,  38,  40,  42,  43,  43,  45,  52,  55,  58, 58\\
 &   58, 59,  63,  63,  71,  75,  85,  86,  86,  99, 102, 102, 141, 153, 161\\
 BAAQMD &  4,  12.5,  15,  15,  20,  20,  20,  25, 170\\
\hline
\end{tabular}
\caption{Carbon monoxide measurements by the refinery and BAAQMD (in ppm). Observations are sorted in ascending order.}
\label{tab:1}
\end{table}

\begin{figure}
\centering
\begin{tabular}
[c]{cc}%
\raisebox{-0cm}{\includegraphics[height=8.8546cm,width=7.5cm]{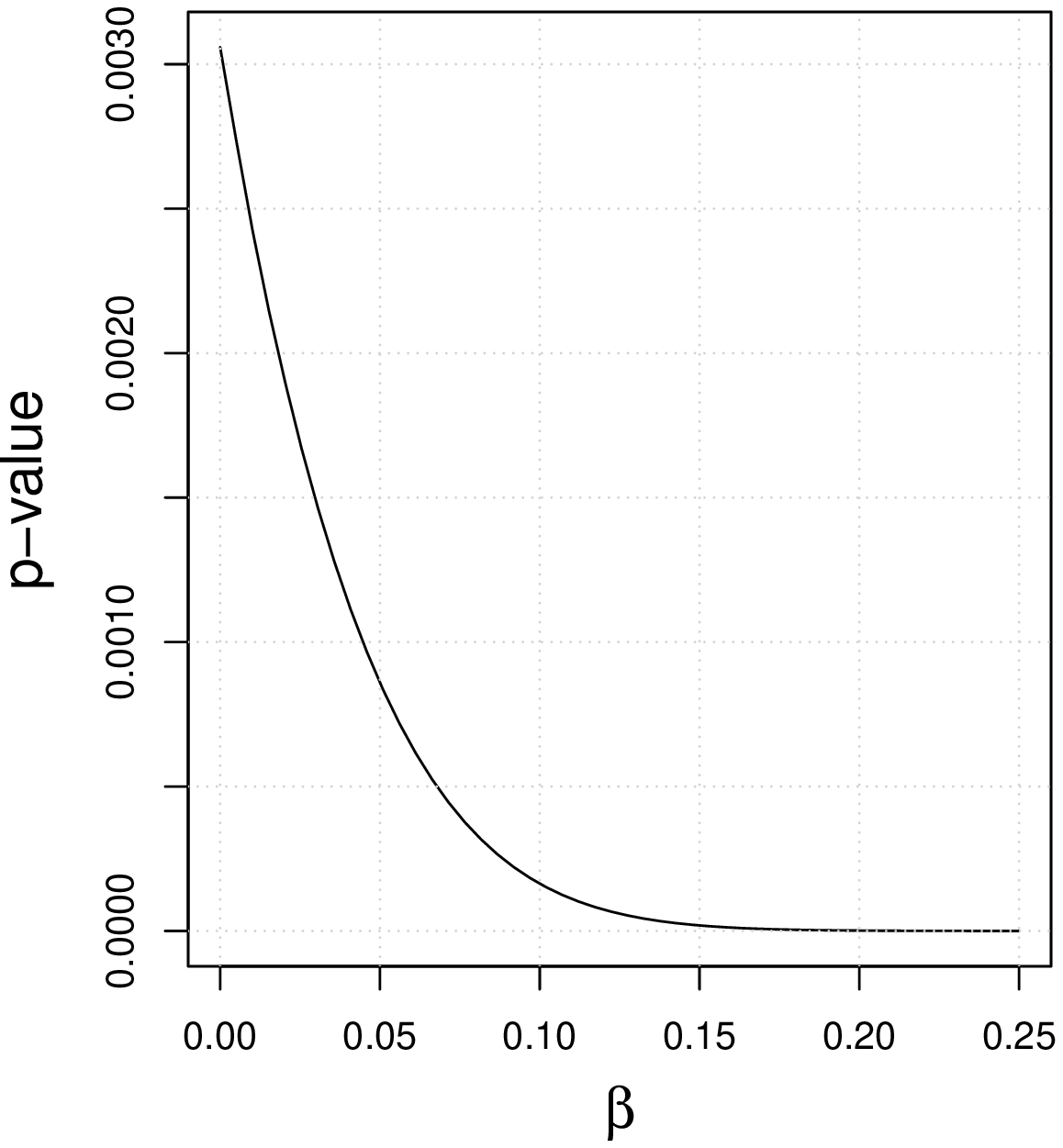}}
&
\raisebox{-0cm}{\includegraphics[height=8.8546cm,width=7.5cm]
{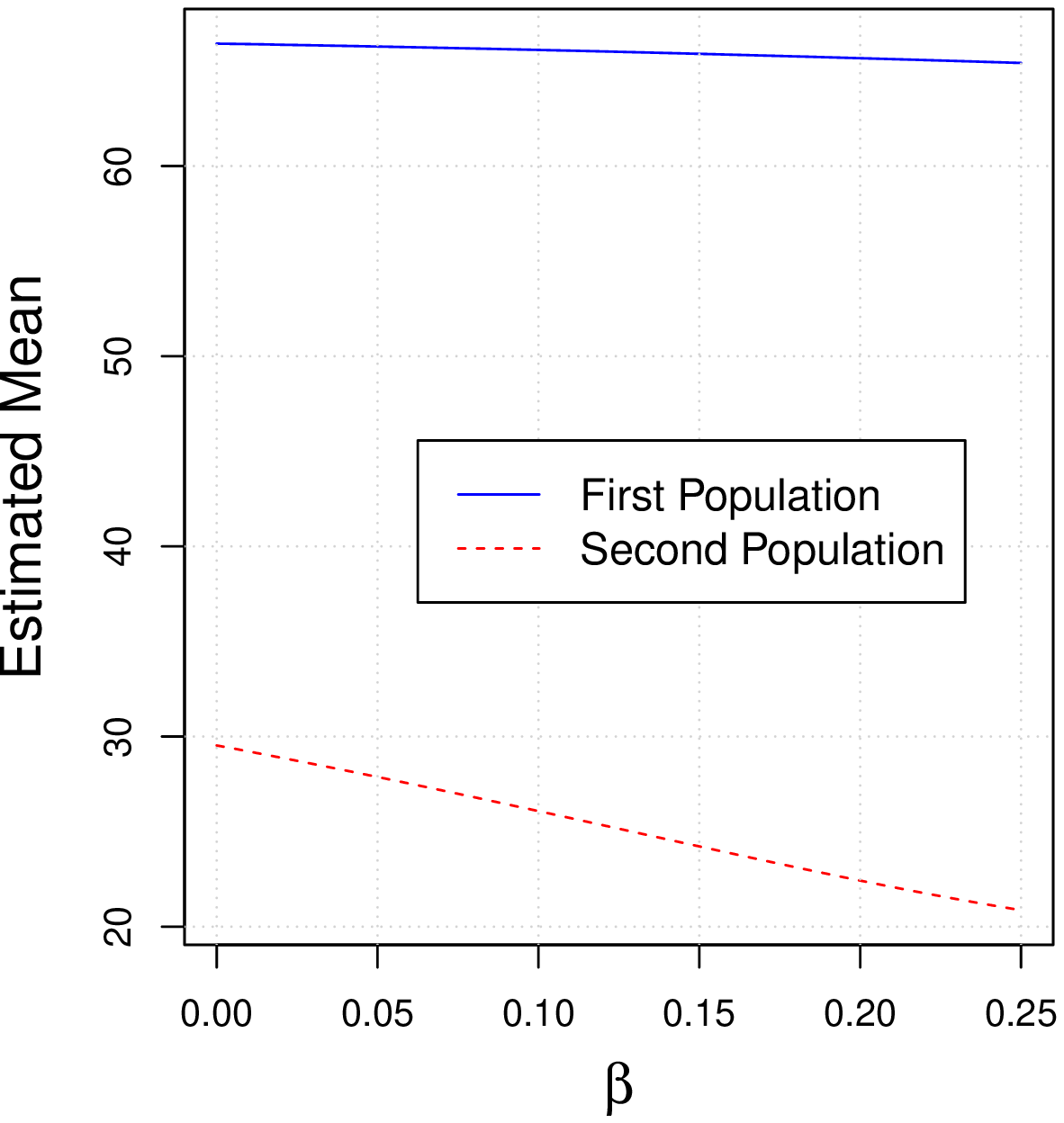}}\\
{\small (a)} & {\small (b)}\\
\raisebox{-0cm}{\includegraphics[height=8.8546cm,width=7.5cm]{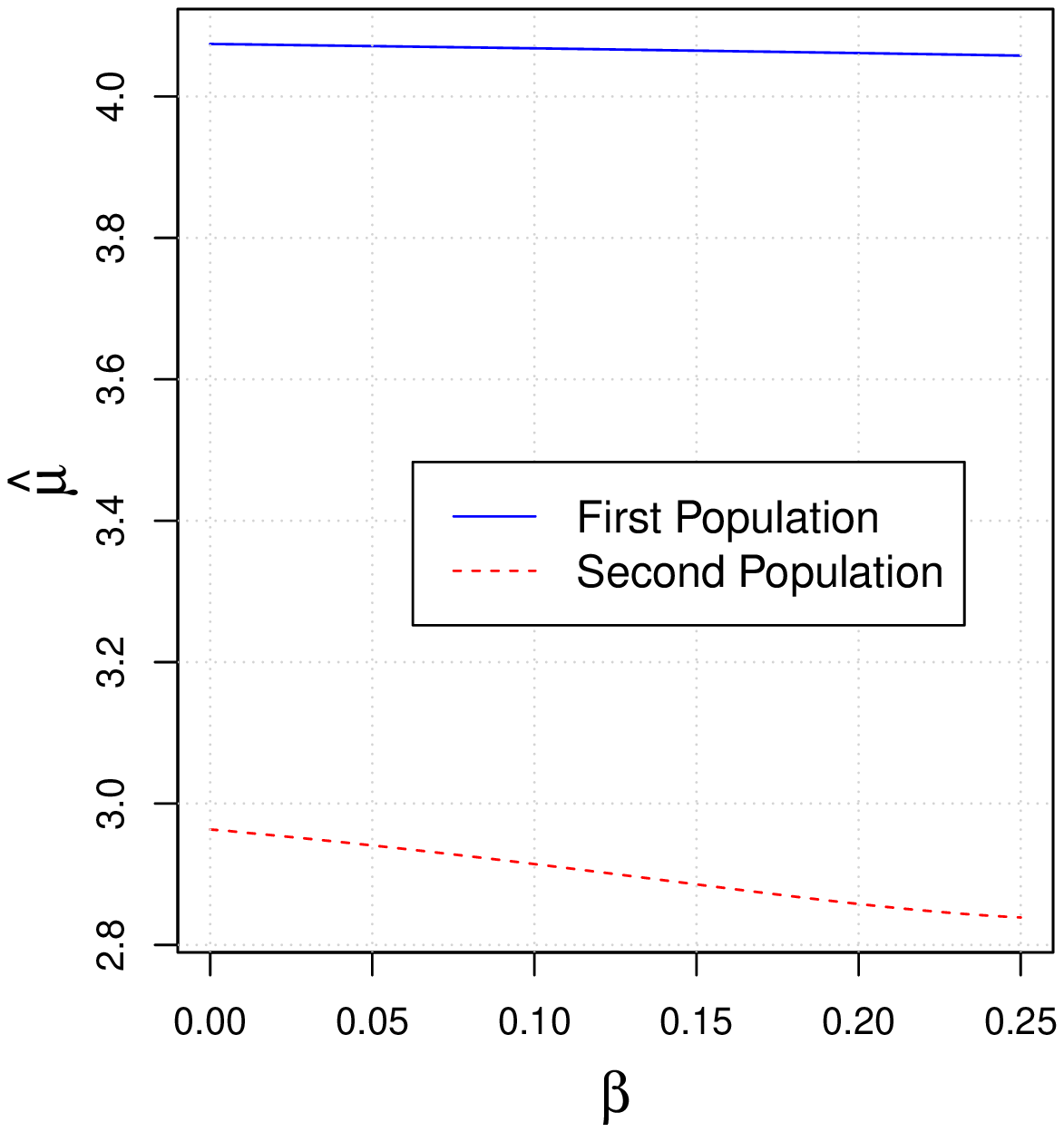}}
&
\raisebox{-0cm}{\includegraphics[height=8.8546cm,width=7.5cm]%
{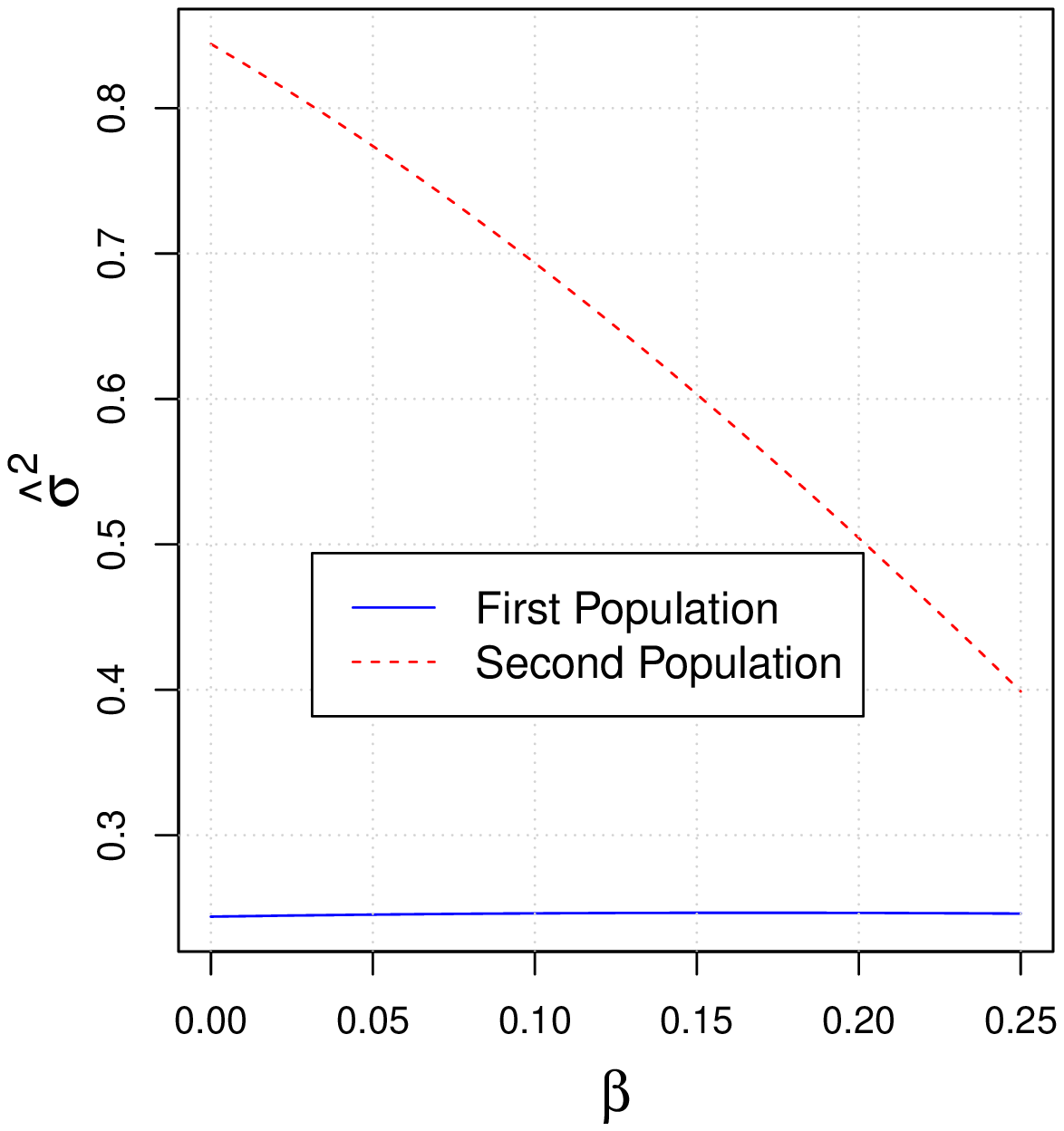}%
}
\\
{\small (c)} & {\small (d)}%
\end{tabular}
\caption{Plots of the (a) $p$-value, (b) estimated means, (c) estimated $\mu$ and (d) estimated $\sigma^2$ for the air quality data. }
\label{fig:air}
\end{figure}%

The values in Table \ref{tab:1} seem to indicate that there is a substantial separation in the two data sets.  In fact, the {\it Data and Story Library} mentioned that the refinery had an incentive to overestimate carbon monoxide emission to set up a baseline at a higher level.  However, there is a big outlier in BAAQMD's measurement at 170. For this reason the mean of the observations from BAAQMD are pushed  closer to the mean of the refinery's measurements. For the full data,  the Z-test, the score test, the bootstrap test and LRT fail to reject the null hypothesis  at the 5\% level. The $p$-values from these tests are 0.0654, 0.2090, 0.2973 and 0.1136, respectively. However, the $p$-values from the Wald-type tests are considerably lower (see Figure \ref{fig:air}(a)). Even DPD(0) generates a $p$-value which would lead to solid rejection in this case; for  $\beta\ge 0.15$ the $p$-values of the Wald-type tests are practically  equal to zero. It is interesting to note that, even if DPD(0) is generally considered to be a non-robust test, it gives a better result than the Z-test, the score test, the bootstrap test and the LRT. The latter tests first transform the observations to the log-scale, which makes a difference with DPD(0).

Figure \ref{fig:air}(b) gives the estimated means of the two populations of the MDPD estimators  for different values of $\beta$. It shows that the difference between the means of two populations increases as $\beta$ increases. For the first population (refinery measurements), the estimated mean remains more or less stable roughly around 66.  On the other hand, the second population (BAAQMD measurement) contains a large outlier, so its estimated mean decreases sharply as $\beta$ increases. As the population mean is a function of $\mu$ and $\sigma^2$, the corresponding estimators also decrease in the second population (see Figures \ref{fig:air}(c) and (d)). If we delete observation 170 from the second population,  the MLEs of $\mu$ and $\sigma^2$ become 2.69 and 0.33, respectively. Figures \ref{fig:air}(c) and (d) show that a large value of $\beta$ considerably eliminates the effect of the outlier. 

To show the effect of an outlier we have moved the outlying value from 1 to 500 (where the original value is 170). Figure \ref{fig_air_outlier} gives the $p$-values of the different tests under this scenario. We have taken three Wald-type tests with $\beta = 0, 0.1$ and 0.2, and compared the results with the LRT, Z-test, score test and bootstrap test. It is clear that Wald-type tests with $\beta =  0.1$ and 0.2  are robust against the outlying observation, and all other tests are outlier sensitive. Although in the original data the $p$-value of DPD(0) was very small, as the value of the outlier increases to 250 or more DPD(0) fails to reject the null hypothesis at 5\% level of significance. It shows that DPD(0) is also a non-robust test, but comparatively less outlier sensitive than the LRT, Z-test, the score test and the bootstrap test.

\begin{figure}
{\begin{center}
\includegraphics[height=8cm,width=12cm]{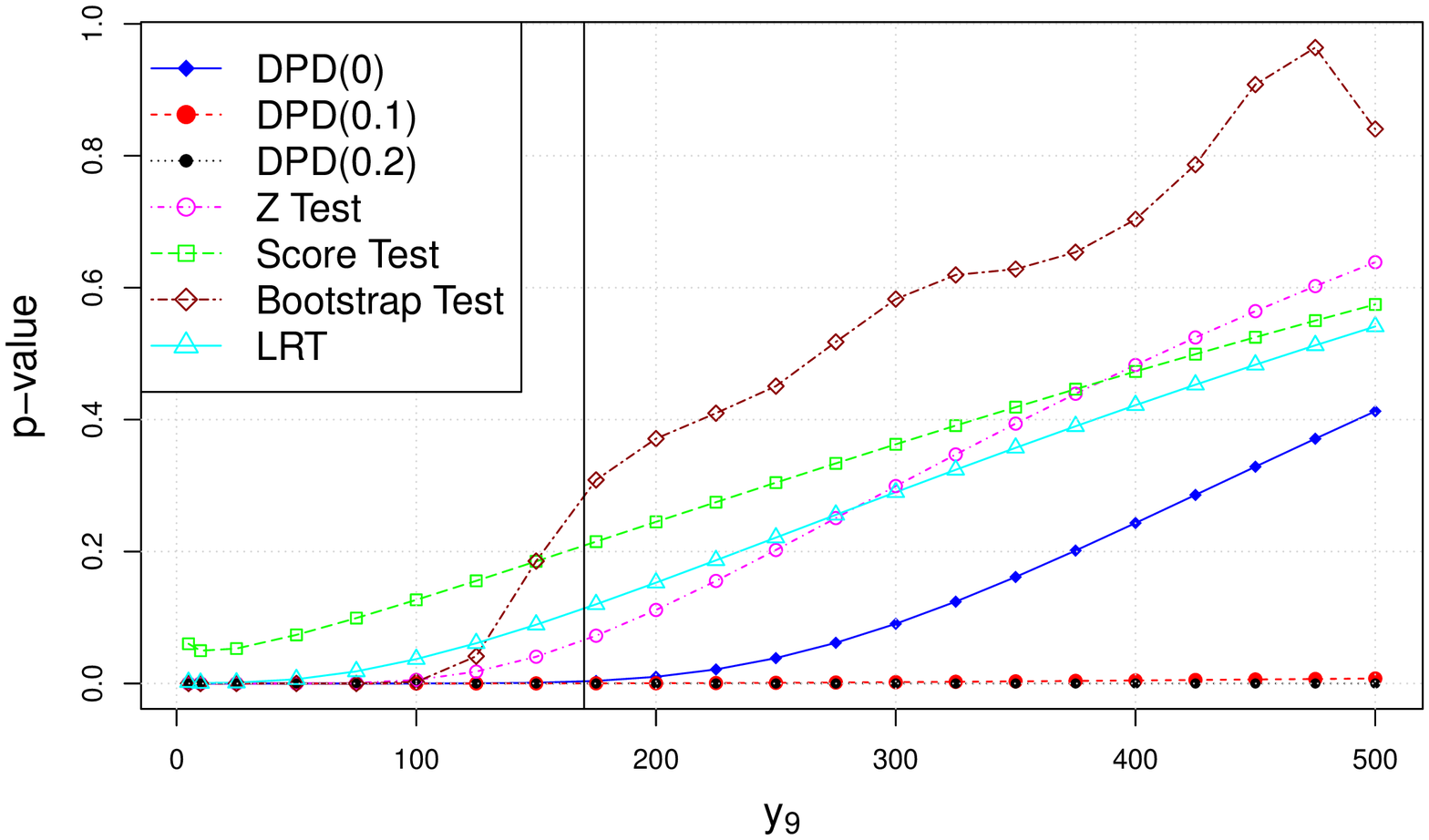}
\end{center}}
\caption{Plot of $p$-values for different values of the outlying observation for the air quality data. The original value of the observation was 170, and the intersection of the different curves with the vertical line at 170 gives the $p$-values for the different methods at the observed data. }
\label{fig_air_outlier}
\end{figure}

\bigskip
\noindent
{\bf Example 2 (Cloud data):} 
Our second example is also from the {\it Data and Story Library} 
and analyzed by
\cite{krishnamoorthy2003inferences}. Table \ref{tab:2} gives the data on the amount of rainfall (in acre-feet) from 52 clouds. Out of them 26 clouds were chosen at random and seeded with silver nitrate. Probability plots indicate that the normal model does not give a reasonable fit whereas the log-normal model fit the data sets very well. Here we are interested to test whether silver nitrate significantly changes the amount of rainfall.  The $p$-values of the Z-test, the score test, the bootstrap test and the LRT for testing the null of the equality of the means (against that they are different)  are 0.1200, 0.1581, 0.0814 and 0.1292, respectively. If we consider a test at 10\% level of significance, then all these $p$-values are on the borderline of rejection, although the bootstrap test is the only one which actually rejects it. On the other hand, Figure \ref{fig:cloud} shows that for large $\beta$ (say $\beta \geq 0.2$) the Wald-type tests produce results which are far from being significant.

\begin{table}
 \begin{tabular}{ l  l }
\hline
 Natural Rainfall &  1.0,    4.9,    4.9,   11.5,   17.3,   21.7,   24.4,   26.1,   26.3,  28.6,   29.0,   36.6,   41.1,   47.3,   \\
 & 68.5, 81.2,   87.0,   95.0, 147.8,  163.0,  244.3,  321.2,  345.5,  372.4,  830.1, 1202.6\\
 Artificial Rainfall &  4.1,    7.7,    17.5,    31.4,    32.7,    40.6,    92.4,    115.3,    118.3 119.0,    129.6,    198.6,    200.7,    \\
 & 242.5, 255.0,    274.7,    274.7,    302.8 334.1,    430.0,    489.1,    703.4,    978.0, 1656.0, \\
 & 1697.8, 2745.6 \\
\hline
\end{tabular}
\caption{The amount of rainfall (in acre-feet) from clouds without  seed and seeded with silver nitrate. Observations are sorted in ascending order.}
\label{tab:2}
\end{table}

\begin{figure}
\centering
\includegraphics[height=8cm,width=10cm]{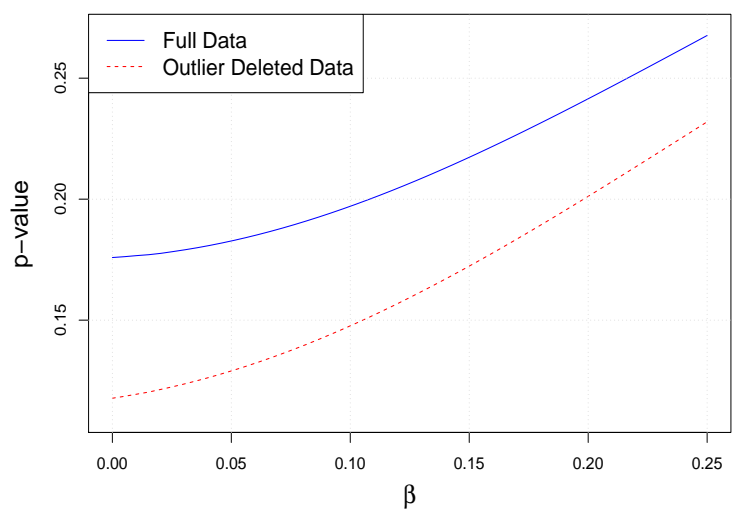}
\caption{ Plot of $p$-values of the Wald-type tests for different values of $\beta$ for the cloud data. In the outlier deleted data the largest observation of the first population is deleted.}
\label{fig:cloud}
\end{figure}

\begin{table}
\centering
\begin{tabular}{lccccccc}
  \hline
   & Z Test & Score Test & Bootstrap & LRT & DPD(0) & DPD(0.1) & DPD(0.2) \\ 
  \hline
 Full Data & 0.1200 & 0.1581 & 0.0814 & 0.1293 & 0.1759 & 0.1971 & 0.2415 \\ 
   Case I & 0.0425 & 0.0820 & 0.0594 & 0.0604 & 0.1177 & 0.1476 & 0.2012 \\ 
   Case II & 0.4782 & 0.4762 & 0.3532 & 0.4645 & 0.4518 & 0.4005 & 0.3830 \\ 
    Case III & 0.2368 & 0.2619 & 0.0850 & 0.2378 & 0.2457 & 0.2521 & 0.2821 \\ 
   \hline
\end{tabular}
\caption{The p-values of different tests for the full cloud data and outlier deleted data -- Case I: the largest value is deleted from the first population, Case II: three largest values are deleted from the second population and Case III: outliers deleted from both populations.}
\label{tab:cloud_outliers}
\end{table}


Table \ref{tab:2} shows that the data set contains a few outliers; in particular, the largest observation of the first population and three largest observations of the second population maybe regarded as outliers. 
If we delete the largest observation from the first population,  the $p$-values of the Z-test, the score test, the bootstrap test and the LRT  drop down to  to 0.0425, 0.0820, 0.0604 and 0.0535, respectively. So, all these tests clearly reject the null hypothesis at 10\% level of significance. On the other hand, the Wald-type tests continue to remain consistent with  the null hypothesis (see Figure \ref{fig:cloud}). Table \ref{tab:cloud_outliers} shows that in other two situations, when the outliers are deleted from the second population and the first population is kept intact, or when the outliers  are deleted from both the populations, all tests, except the bootstrap test, clearly fail to reject the null hypothesis. On the whole, the Wald-type tests with moderately large values of $\beta$ are highly stable and lead to the same conclusion in each of the four scenarios considered in Table \ref{tab:cloud_outliers}.

\section{Conclusion} \label{conclusion}
The log-normal distribution is widely used in modeling biomedical and medical data. The statistical analysis for comparing the means of two independent log-normal distributions is of real interest for the researchers. In the literature there are several tests which deal with this problem, but they are often non-robust in the presence of outliers and model misspecification. In this paper we have presented a robust test based on the minimum density power divergence estimator. The asymptotic distribution and the power functions for the fixed as well as the contiguous alternatives are presented. The robustness property of the test is also theoretically established. An extensive simulation work is conducted to verify the theoretical results, and the test is compared with some existing methods. This comparison demonstrates that our method outperforms the existing tests in the presence of outliers, while at the same time it gives very comparative results when the model is correctly specified. Two real-life examples are presented where, due to the presence of outliers, the existing tests give a misleading information, but our test produces a reliable result. 

\bigskip

\noindent
\textbf{Acknowledgement. } This research is partially supported by Grant MTM2015-67057-P from Ministerio de Economia y Competitividad (Spain). 

 \bibliography{LN_ReferenceNew}

\end{document}